\documentclass{article}

\usepackage[preprint]{neurips_2026}

\usepackage[utf8]{inputenc} % allow utf-8 input
\usepackage[T1]{fontenc}    % use 8-bit T1 fonts
\usepackage{hyperref}       % hyperlinks
\usepackage{url}            % simple URL typesetting
\usepackage{booktabs}       % professional-quality tables
\usepackage{amsfonts}       % blackboard math symbols
\usepackage{nicefrac}       % compact symbols for 1/2, etc.
\usepackage{microtype}      % microtypography
\usepackage{xcolor}         % colors
\usepackage{amsthm}
\usepackage{amsmath}
\usepackage{mathtools}
\usepackage{algorithm}
\usepackage{algorithmic}
\newtheorem{theorem}{Theorem}
\newtheorem{lemma}{Lemma}
\newtheorem{definition}{Definition}
\newtheorem{proposition}{Proposition}

\newtheorem{assumption}{Assumption}

\title{Bayesian Persuasion with a Risk-Conscious Receiver}

  \author{
Yujing Chen \\
School of Mathematical Sciences, Peking University \\
\texttt{yujingchen@stu.pku.edu.cn}
}

\begin{document}

\maketitle

\begin{abstract}

We study Bayesian persuasion when the receiver evaluates actions by reward-side Conditional Value-at-Risk (CVaR) rather than expected utility.  CVaR preferences break the standard action-based direct-recommendation reduction: merging signals that recommend the same action can change the receiver's tail-risk ranking and destroy incentive compatibility.  
We show that this failure does not imply intractability in the explicit finite-state model.  Each CVaR action value is max-affine in the posterior, and refining recommendations by the active affine piece yields an active-facet revelation principle and an exact polynomial-size linear program.  We further identify a representation boundary: listed polyhedral risks remain tractable by the same LP, whereas succinctly represented facet families make exact persuasion NP-hard.  Finally, we give a finite-precision approximation scheme for risk preferences determined by finitely many stable posterior statistics.
\end{abstract}

\section{Introduction}

In many strategic communication settings, a sender influences a receiver by selectively revealing information about an uncertain state of the world.  Bayesian persuasion, introduced by \citet{kamenica_bayesian_2011}, gives a clean formal model for this problem and has found applications in advertising, security, regulation, and resource allocation \citep{badanidiyuru_targeting_2018,rabinovich_information_2015,yang_improving_2013}.  The classical theory usually assumes that the receiver is risk-neutral, so her payoff from an action is linear in the posterior belief.  This linearity is mathematically powerful: it supports the revelation principle and leads to tractable linear-programming formulations.

Many decision environments violate this premise.  A regulator evaluating a model release, a platform deciding whether to surface an automated recommendation, or an operator deciding whether to trust an alert may care less about average performance than about rare but severe failures.  In such settings, the receiver may reject an action that looks favorable in expectation if a small posterior tail contains unacceptable downside risk.  The information-design problem is then not only to persuade an expected-utility receiver, but to communicate in a way that remains persuasive under tail-risk evaluation.

We capture this form of caution through Conditional Value-at-Risk (CVaR), a standard risk measure that focuses on the worst tail of an outcome distribution \citep{artzner_coherent_1999,rockafellar_optimization_2000,rockafellar_conditional_2002,acerbi_coherence_2002,jorion_value_2006,filippi_conditional_2020}.  Under CVaR preferences, the receiver's utility is nonlinear in posterior beliefs.  As a result, merging two signals that separately induce the same action can destroy incentive compatibility, so direct recommendation schemes are no longer without loss.  Recent work has developed geometric insight for risk-conscious receivers \citep{anunrojwong_persuading_2024a}; this paper gives a computational account of the CVaR case.

\paragraph{Our contributions.}
The paper studies how tail-risk preferences alter the reduced form and computational structure of Bayesian persuasion.  The first observation is that CVaR, and more generally nonlinear risk preferences, can invalidate the usual action-based revelation argument. Under expected utility, the receiver's incentive constraints are linear in the posterior, so signals that induce the same action can be merged.  Under CVaR, merging changes the distribution of outcomes under the posterior and may change the relevant lower tail.  An action that is optimal at two posteriors need not remain optimal at their convex combination.

The main positive result is an active-facet replacement for direct revelation.  In a finite state space, the CVaR value of each action is max-affine in the posterior.  A signal can therefore be indexed by an action \(a\) together with an active affine piece \(\ell\) supporting the receiver's CVaR value.  Conditional on this action-facet pair, incentive compatibility is described by linear inequalities.  This gives an active-facet revelation principle and an exact linear program with polynomially many variables and constraints in the explicit finite-state model.

The same argument also clarifies the role of representation.  CVaR is one instance of an explicitly listed polyhedral risk preference,
$
    \rho(\mu,a)=\max_{\ell\in\mathcal L_a}\langle c_{a,\ell},\mu\rangle .
$
Whenever the affine pieces are part of the input, the active-facet LP remains polynomial in the listed representation.  By contrast, if the affine pieces are specified only implicitly by a succinct combinatorial family, the representation can hide an NP-hard search problem.  We prove such a hardness result for succinct polyhedral risk preferences with an explicit finite state space and two receiver actions.

The approximation results address this representation gap.  They are not needed to solve the tabular CVaR problem, but they give a finite-precision method for settings in which the exact facet description is unavailable or too large.  Under a finite-statistic access condition, where the receiver's risk value depends Lipschitz-continuously on finitely many posterior statistics
and the induced statistic cells can be enforced, a discretized state-contingent LP gives uniform \(\epsilon\)-incentive compatibility. A positive margin condition converts this approximate guarantee into strict incentive compatibility, and a local active-facet refinement shows how the grid complexity can depend on the number of facets that are active near the relevant posterior region rather than on the global facet count.

\subsection{Related Work}
We give a short overview here and defer a fuller discussion to
Appendix~\ref{app:related_work}. Bayesian persuasion was introduced by
\citet{kamenica_bayesian_2011}; see \citet{bergemann_information_2019} for a
survey. Classical tractability relies on expected-utility linearity, which
supports concavification and direct recommendations. We ask what remains when
the receiver evaluates actions by tail risk.

The closest conceptual predecessor is work on risk-conscious receivers
\citep{anunrojwong_persuading_2024a}. We focus on CVaR and show that its
finite max-affine representation yields an active-facet revelation principle
and an exact LP in the explicit finite-state model. This connects to
algorithmic persuasion, where LP tractability is well understood under
expected utility \citep{xu_exploring_2015,dughmi_algorithmic_2021}. The
explicit-versus-succinct distinction used here follows a standard
representation view in complexity theory: compact encodings can change the
complexity of graph and combinatorial problems
\citep{galperin_succinct_1983,papadimitriou_note_1986,wagner_complexity_1986}.
In our setting, listed polyhedral risks remain tractable, while succinctly
represented facet families can be hard.

The finite-precision part uses the sampling intuition behind sparse
approximations in games \citep{althofer_sparse_1994,lipton_playing_2003}, but
the approximated objects are posterior statistics determining the receiver's
risk value rather than mixed strategies or payoff vectors.
\section{The General Risk-Conscious Persuasion Model}
\label{sec:model}

We first describe a persuasion model in which the receiver's preference over actions may be nonlinear in posterior beliefs. This general formulation keeps the role of the risk functional separate from the information structure; CVaR is introduced in the next section.

\subsection{The Environment and Information Structure}
We consider a game with a single Sender ($S$) and a single Receiver ($R$). The uncertainty is a state $\omega$ drawn from a finite state space $\Omega = \{\omega_1, \dots, \omega_m\}$. The Sender and Receiver share a common prior $\mu_0 \in \operatorname{int}(\Delta(\Omega))$, where $\Delta(\Omega)$ is the $(m-1)$-dimensional probability simplex over $\Omega$. The Receiver chooses from a finite action set $\mathcal{A} = \{a_1, \dots, a_n\}$. 
Throughout the paper, we write \(n=|\mathcal A|\).

The Sender possesses an informational advantage, observing the realization of the true state $\omega$. To influence the Receiver's action, the Sender commits to a \textit{signaling scheme} (or information structure) $\pi$. 
\begin{definition}[Signaling Scheme]
A signaling scheme is a tuple $(\mathcal{S}, \pi)$, where $\mathcal{S}$ is a finite space of signal realizations, and $\pi: \Omega \to \Delta(\mathcal{S})$ is a mapping from states to distributions over signals. We denote $\pi(s|\omega)$ as the probability of sending signal $s \in \mathcal{S}$ conditional on the realized state being $\omega$.
\end{definition}

Upon observing a signal $s \in \mathcal{S}$, the Receiver updates her belief about the state of the world using Bayes' rule. The posterior belief $\mu_s \in \Delta(\Omega)$ is given by:
\begin{equation}
    \mu_s(\omega) = \frac{\pi(s|\omega)\mu_0(\omega)}{\sum_{\omega' \in \Omega} \pi(s|\omega')\mu_0(\omega')}, \quad \forall \omega \in \Omega.
\end{equation}
The unconditional probability of observing signal $s$ is $\mathbb{P}(s) = \sum_{\omega'} \pi(s|\omega')\mu_0(\omega')$. Thus a signaling scheme induces a distribution of posterior beliefs $\tau \in \Delta(\Delta(\Omega))$. Bayesian updating imposes the Bayes-plausibility condition that the mean posterior equals the prior:
\begin{equation}
    \sum_{s \in \mathcal{S}} \mathbb{P}(s) \mu_s = \mu_0.
\end{equation}

\subsection{Generalized Risk Preferences}
In the canonical Bayesian persuasion model, the Receiver is assumed to be risk-neutral, maximizing an expected utility $\mathbb{E}_{\mu_s}[u^R(\omega, a)]$ which is strictly linear in the posterior belief $\mu_s$. 

In the generalized model, the Receiver evaluates actions through a risk functional $\rho: \Delta(\Omega) \times \mathcal{A} \to \mathbb{R}$. Given posterior $\mu_s$, the Receiver chooses
\begin{equation} \label{eq:receiver_general_obj}
    a^*(s) \in \operatorname{argmax}_{a \in \mathcal{A}} \rho(\mu_s, a).
\end{equation}

Throughout, incentive compatibility is understood in the weak sense: a
recommended action must belong to the receiver's best-response set.  When the
sender's value is evaluated at a posterior with multiple receiver best
responses, we use sender-favorable tie-breaking.

The functional $\rho$ may depend nonlinearly on $\mu_s$, allowing the model to represent preferences such as tail-risk sensitivity or ambiguity aversion. The Sender is risk-neutral with utility $v: \Omega \times \mathcal{A} \to \mathbb{R}$ and solves the following bilevel problem:
\begin{equation}\label{eq:general_sender_problem}
    \begin{aligned}
        \max_{\pi:\Omega \to \Delta(\mathcal{S})} \quad & \sum_{\omega \in \Omega} \sum_{s \in \mathcal{S}} \mu_0(\omega) \pi(s|\omega) v(\omega, a^*(s))\\
    \text{s.t.} \quad & a^*(s) \in \operatorname{argmax}_{a \in \mathcal{A}} \rho(\mu_s, a) \quad \forall s \in \mathcal{S}\\
    &\sum_{s \in \mathcal{S}} \pi(s|\omega) = 1 \quad \forall \omega \in \Omega\\
    &\pi(s|\omega)  \ge 0 \quad \forall s \in \mathcal{S}, \omega \in \Omega
    \end{aligned}
\end{equation}

\subsection{Failure of Action-Based Revelation}
\label{subsec:revelation_failure}

The nonlinear dependence of \(\rho(\mu,a)\) on posterior beliefs changes the
standard reduced form of Bayesian persuasion.  Under expected utility, the
receiver's payoff from each action is linear in \(\mu\).  If two signals both
make the same action optimal, their merger preserves optimality, and the sender
can restrict attention to direct recommendations indexed by actions.

This argument fails for general risk functionals.  The set
$
    \{\mu\in\Delta(\Omega): a\in\arg\max_{a'}\rho(\mu,a')\}
$
need not be convex.  Hence two posteriors that both induce action \(a\) may
have a convex combination at which another action is preferred.  Merging
signals that recommend the same action can therefore destroy incentive
compatibility.

The failure of action-based revelation does not by itself imply computational
hardness.  It means that a different reduced form is needed.  For arbitrary
nonlinear risk preferences, such a reduced form may not be available.  The
next section shows that CVaR has enough additional structure to recover
tractability: its value is max-affine in the posterior, and signals can be
refined by the active affine facet supporting the recommended action.

For completeness, Appendix~\ref{app:general_cardinality_proof} records a
standard finite-support bound for arbitrary risk preferences.  This bound is a
geometric background result and is not used in the active-facet LP construction.

\section{Tractability and Optimization under CVaR}
\label{sec:cvar_optimization}

We now specialize the receiver's risk functional to Conditional Value-at-Risk
(CVaR).  The role of this specialization is structural.  Although CVaR is
nonlinear in posterior beliefs, it is polyhedral on every finite state space:
for each action, the receiver's value is the maximum of finitely many affine
functions.  This max-affine structure is what makes an exact optimization
result possible.

\subsection{CVaR-Conscious Preferences and Piecewise Linearity}
For a given confidence level $r \in (0, 1)$, the Receiver evaluates the CVaR of the utility $u(\omega, a)$, which captures the expectation of the worst $r$-quantile outcomes. The risk functional is defined as:
\begin{equation}
    \rho(\mu, a) \coloneqq \operatorname{CVaR}_r^\mu(u(\cdot, a)) = \sup_{b \in \mathbb{R}} \left( b - \frac{1}{r} \mathbb{E}_{\omega \sim \mu} [(b - u(\omega, a))^+] \right).
\end{equation}

Although CVaR is nonlinear in beliefs, it has a simple finite-facet representation on a finite state space.

\begin{lemma}[Finite-Facet Representation of CVaR]\label{lemma:cvar_pwl}
For any fixed action $a\in\mathcal A$, the function
$
f_a(\mu)=\operatorname{CVaR}_r^\mu(u(\cdot,a))
$
is continuous, convex, and piecewise linear in $\mu$. More precisely, there exists a finite set $\mathcal L_a$ with $|\mathcal L_a|\le |\Omega|$ and vectors $c_{a,l}\in\mathbb R^{|\Omega|}$ such that
\begin{equation}
f_a(\mu)=\max_{l\in\mathcal L_a}\langle c_{a,l},\mu\rangle.
\end{equation}
If $|u(\omega,a)|\le C_R$ for all $(\omega,a)$, then the coefficients may be chosen to satisfy
\begin{equation}
\|c_{a,l}\|_\infty\le C_R\left(1+\frac2r\right).
\end{equation}
\end{lemma}
Proof in Appendix~\ref{app:cvar_proof}.

\subsection{Active-Facet Revelation and Exact Optimization}
\label{subsec:active_facet_lp}

Lemma~\ref{lemma:cvar_pwl} shows that the receiver's CVaR value is nonlinear
but finitely piecewise linear in posterior beliefs. This finite-facet
structure gives the right replacement for the usual action-based revelation
principle. Under expected utility, signals that recommend the same action can
be merged without changing incentives. Under CVaR, such a merge may change
the active tail event and therefore change the receiver's ranking of actions.
Thus the appropriate reduced form is not indexed by actions alone, but by
actions together with their active CVaR facets.

For each action \(a\in\mathcal A\), let \(\mathcal L_a\) be the finite
index set from Lemma~\ref{lemma:cvar_pwl}, so that
\[
    f_a(\mu)
    =
    \operatorname{CVaR}_r^\mu(u(\cdot,a))
    =
    \max_{\ell\in\mathcal L_a}
    \langle c_{a,\ell},\mu\rangle .
\]
For every pair \((a,\ell)\), define the refined incentive region
\begin{equation}
    P_{a,\ell}
    :=
    \left\{
    \mu\in\Delta(\Omega):
    \langle c_{a,\ell},\mu\rangle
    \ge
    \langle c_{a',\ell'},\mu\rangle
    \quad
    \text{for all }a'\in\mathcal A,\ \ell'\in\mathcal L_{a'}
    \right\}.
\end{equation}
If \(\mu\in P_{a,\ell}\), then facet \(\ell\) is active for action \(a\),
and action \(a\) is a weak best response for the receiver. Since all
comparisons are linear, each \(P_{a,\ell}\) is a polytope.

This yields a refined revelation principle. Although signals cannot in
general be merged by action alone, they can be grouped by the pair consisting
of the recommended action and an active CVaR facet. 

\begin{proposition}[Active-Facet Revelation Principle]

\label{prop:active_facet_revelation}

For every incentive-compatible signaling scheme in the finite-state CVaR
persuasion problem, there exists another incentive-compatible signaling scheme
with the same sender value whose signals are indexed by pairs
\((a,\ell)\), where \(a\in\mathcal A\) is the recommended action and
\(\ell\in\mathcal L_a\) is an active CVaR facet at the induced posterior.
Moreover, the posterior associated with signal \((a,\ell)\) belongs to
\(P_{a,\ell}\).
\end{proposition}

Proof in Appendix~\ref{app:active_facet_revelation}.

The resulting problem is
linear in the joint distribution over states and refined recommendations.

\begin{theorem}[Exact LP for finite-state CVaR persuasion]
\label{thm:active_facet_lp}
Consider a finite-state Bayesian persuasion problem with a single CVaR
receiver, finite action set \(\mathcal A\), explicitly given sender and
receiver payoffs, prior \(\mu_0\), and rational risk level \(r\in(0,1)\).
Let \(m=|\Omega|\), \(n=|\mathcal A|\), and
\[
    L:=\sum_{a\in\mathcal A}|\mathcal L_a|.
\]
Then an optimal incentive-compatible signaling scheme can be computed by a
linear program with \(mL\) nonnegative joint-mass variables and at most
\(m+L^2\) linear constraints, apart from nonnegativity. Since
\(L\le nm\), this is \(O(nm^2)\) variables and \(O(n^2m^2)\) incentive
constraints. In particular, the problem is solvable in time polynomial in
\(n\), \(m\), and the input bit length.
\end{theorem}

Proof in Appendix~\ref{app:active_facet_lp}.

The theorem clarifies the role of Lemma~\ref{lemma:cvar_pwl}. CVaR does
break the classical action-only revelation argument, because the active
tail facet may change under merging. But once the active facet is kept as
part of the recommendation type, the incentive regions become polytopes and
the sender's optimization problem becomes an ordinary linear program.

The appendices contain two complementary views of this tractability result.
Appendix~\ref{app:low_dim} works out the binary state-action case explicitly,
where the geometry of CVaR incentives can be visualized directly. Appendix~\ref{app:poly_proof} gives a fixed-dimensional cell-enumeration argument.
Both are subsumed by the active-facet LP in Theorem~\ref{thm:active_facet_lp},
but they help explain why the correct refinement is by action-facet pairs
rather than by actions alone.
\section{The Polyhedral Boundary: Explicit Tractability and Succinct Hardness}
\label{sec:polyhedral_boundary}

The active-facet LP reveals that the tractability of finite-state CVaR
persuasion is not specific to the variational formula of CVaR.  The essential
property is max-affinity: once the receiver's value is represented as the
maximum of affine functions of the posterior, each signal can be refined by the
affine piece supporting the recommended action.  This section makes the
corresponding representation boundary precise.  If the affine pieces are
explicitly listed, the same LP argument remains polynomial.  If they are given
only through a succinct combinatorial description, the hidden facet structure
can encode NP-hard search problems.

\subsection{Explicitly Listed Polyhedral Risk Preferences}
\label{subsec:explicit_polyhedral_risk}

We first record the direct extension of the CVaR result to general polyhedral
risk preferences with listed facets.

\begin{definition}[Explicit polyhedral risk preference]
\label{def:explicit_polyhedral_risk}
A risk preference is \emph{explicit polyhedral} if, for every receiver action
\(a\in\mathcal A\),
\[
    \rho(\mu,a)
    =
    \max_{\ell\in\mathcal L_a}\langle c_{a,\ell},\mu\rangle,
\]
where the finite index set \(\mathcal L_a\) and all coefficient vectors
\(c_{a,\ell}\in\mathbb R^{|\Omega|}\) are given as part of the input.  Let
$
    L:=\sum_{a\in\mathcal A}|\mathcal L_a|
$
denote the total number of listed affine pieces.
\end{definition}

CVaR is the leading example of this class.  By
Lemma~\ref{lemma:cvar_pwl}, each action has at most \(|\Omega|\) affine pieces,
and these pieces can be constructed from the payoff table.  The next theorem
states the general form of the active-facet LP.  Its running time is polynomial
in the number of states, actions, and listed affine pieces.

\begin{theorem}[Exact LP for explicit polyhedral risk persuasion]
\label{thm:explicit_polyhedral_lp}
Consider a finite-state Bayesian persuasion problem with finite action set
\(\mathcal A\), prior \(\mu_0\), sender payoff \(v\), and an explicit
polyhedral receiver risk preference as in
Definition~\ref{def:explicit_polyhedral_risk}.  Let \(m=|\Omega|\) and
\(L=\sum_{a\in\mathcal A}|\mathcal L_a|\).  Then an optimal
incentive-compatible signaling scheme can be computed by a linear program with
\(mL\) nonnegative joint-mass variables and at most \(m+L^2\) linear
constraints, apart from nonnegativity.  Hence the problem is solvable in time
polynomial in \(m\), \(L\), and the input bit length.
\end{theorem}

Proof in Appendix~\ref{app:explicit_polyhedral_lp}.

This theorem is a representation statement.  Polyhedrality alone is not a
source of computational hardness: if all facets are listed, they can be treated
as refined recommendation types.  Any hardness result must therefore rely on a
different representation model, in which the relevant affine pieces are hidden,
implicit, or exponentially many.

\subsection{Succinct Polyhedral Risk and Computational Hardness}
\label{subsec:succinct_polyhedral_hardness}

We next consider such a representation model.  A succinct polyhedral risk
preference still has a max-affine form, but the affine pieces are specified by
a compact combinatorial rule rather than by an explicit list.  Expanding all
facets would recover the LP above, but the expanded formulation may have
exponential size.  Moreover, identifying a useful hidden facet can itself be a
hard combinatorial problem.

A polyhedral risk functional is given succinctly if, for each action \(a\),
\[
    \rho(\mu,a)
    =
    \max_{\ell\in\mathcal L_a}\langle c_{a,\ell},\mu\rangle,
\]
but the index family \(\mathcal L_a\) is specified by a polynomial-size
combinatorial description rather than by enumeration.  

In the reduction, the hidden affine pieces are indexed by the \(K\)-cliques of
an input graph, while all payoffs and affine coefficients lie in \([0,1]\).

\begin{theorem}[NP-hardness for succinct polyhedral risk]
\label{thm:succinct_polyhedral_np_hard}
Deciding whether a signaling scheme can achieve sender value at least a given
threshold is NP-hard for succinctly represented polyhedral risk preferences,
even with an explicit finite state space, two receiver actions, and a uniform
prior.
\end{theorem}

Proof in Appendix~\ref{app:succinct_polyhedral_np_hard}.

Together, Theorems~\ref{thm:explicit_polyhedral_lp} and
\ref{thm:succinct_polyhedral_np_hard} identify a representation boundary.
Max-affine risk preferences are tractable when their affine pieces are listed,
but the same structure can be computationally hard when those pieces are given
only implicitly.  The obstruction is therefore not nonlinearity per se, but
the representation of the receiver's polyhedral risk landscape.

\section{Posterior Discretization for Finite-Statistic Risk Preferences}
\label{sec:qptas}

The hardness result above is an exact worst-case statement for arbitrary
succinct descriptions of polyhedral risk.  We now turn to a finite-precision
setting. The approximation problem does not require recovering every hidden facet of the
receiver's risk landscape; it only requires preserving the posterior statistics
that determine the receiver's evaluation. This leads to approximate incentive compatibility
rather than exact incentive compatibility.  For implicit statistic families,
the relevant input model is the cell-separation model in
Assumption~\ref{assump:finite_statistic_access}.

\begin{assumption}[Finite-statistic cell access]
\label{assump:finite_statistic_access}
There exist statistic functions
\[
    g_1,\ldots,g_D:\Omega\to\mathbb R,
    \qquad |g_j(\omega)|\le C_g,
\]
and maps \(\Psi_a:\mathbb R^D\to\mathbb R\), one for each action
\(a\in\mathcal A\), such that
\begin{equation}\label{eq:finite_statistic_risk}
    \rho(\mu,a)
    =
    \Psi_a\bigl(
        \langle g_1,\mu\rangle,\ldots,\langle g_D,\mu\rangle
    \bigr).
\end{equation}
Each \(\Psi_a\) is \(L_\Psi\)-Lipschitz in the \(\ell_\infty\)-norm.  The
algorithm can evaluate \(\rho(\bar\mu,a)\) at grid posteriors and can enforce
the statistic cells used below: given a grid center \(\bar\mu\), a tolerance
\(\eta\), and a candidate posterior \(\mu\), it can either certify
\begin{equation}\label{eq:cell_access}
    \max_{j\in[D]}
    |\langle g_j,\mu-\bar\mu\rangle|
    \le \eta
\end{equation}
or return a violated statistic.  
This condition is satisfied when the statistics are explicitly listed; more
generally, it may be imposed through a polynomial-time separation oracle for
the statistic cells.
\end{assumption}
This access condition is used only for the finite-precision scheme; the exact
LPs for CVaR and explicitly listed polyhedral risks do not rely on it.

For max-affine risks, the statistics can be chosen as the affine pieces.  If
\[
    \rho(\mu,a)=\max_{\ell\in\mathcal L_a}
    \langle c_{a,\ell},\mu\rangle,
\]
then set \(g_{a,\ell}=c_{a,\ell}\).  In this case
\(D=\sum_a|\mathcal L_a|\) and \(L_\Psi=1\).  For CVaR,
Lemma~\ref{lemma:cvar_pwl} gives \(D\le mn\) and
\(C_g\le M_R(1+2/r)\).  For succinct polyhedral risks, \(D\) may be
exponentially large.  The bounds below depend on \(\log D\), while the LP
implementation relies on the cell-access condition in
Assumption~\ref{assump:finite_statistic_access}.

For the discretization bounds, write
\begin{equation}
    B_\rho:=C_gL_\Psi .
\end{equation}

\paragraph{Sampling finite statistics.}
The discretization uses the standard empirical approximation argument for
finite families of bounded linear statistics.  A posterior need not be
approximated in every coordinate of the simplex; it is enough to approximate
the statistics through which the receiver's risk value changes.  The
state-contingent LP then enforces these statistic cells directly.

\subsection{Discretization by \(k\)-Uniform Distributions}
The belief simplex \(\Delta(\Omega)\) is continuous, so we replace it with the
finite grid of \(k\)-uniform distributions. For \(k\in\mathbb N\), let
\begin{equation}
    \mathcal D_k
    =
    \left\{
    \mu\in\Delta(\Omega):
    \mu(\omega)\in\left\{0,\frac1k,\ldots,1\right\}
    \text{ for every }\omega
    \right\}.
\end{equation}
Its cardinality is
\[
    |\mathcal D_k|
    =
    \binom{m+k-1}{k}
    \le m^k.
\]

The grid must approximate receiver values for all actions.  Under the
finite-statistic representation, it is enough to approximate the statistics
\(\langle g_j,\mu\rangle\) uniformly.

\begin{lemma}[Uniform finite-statistic approximation by \(k\)-uniform posteriors]
\label{lemma:approx}
Suppose Assumption~\ref{assump:finite_statistic_access} holds.  For any
posterior \(\mu\in\Delta(\Omega)\) and any \(\epsilon_R>0\), if
$
    k
    \ge
    \frac{2B_\rho^2}{\epsilon_R^2}
    \log(2D),
$
then there exists a \(k\)-uniform posterior \(\bar\mu\in\mathcal D_k\) such
that
\begin{equation}\label{eq:stat_approx}
    \max_{j\in[D]}
    \left|
    \langle g_j,\mu-\bar\mu\rangle
    \right|
    \le
    \frac{\epsilon_R}{L_\Psi}.
\end{equation}
Consequently,
\begin{equation}\label{eq:risk_value_approx}
    \max_{a\in\mathcal A}
    |\rho(\mu,a)-\rho(\bar\mu,a)|
    \le \epsilon_R.
\end{equation}
\end{lemma}

\begin{proof}
Sample \(X_1,\ldots,X_k\) independently from \(\mu\), and let \(\bar\mu\) be
the empirical distribution.  For a fixed statistic \(g_j\), Hoeffding's
inequality gives
\[
    \Pr\left(
    \left|
    \langle g_j,\bar\mu-\mu\rangle
    \right|
    \ge
    \frac{\epsilon_R}{L_\Psi}
    \right)
    \le
    2\exp\left(
    -\frac{k\epsilon_R^2}{2B_\rho^2}
    \right).
\]
Taking a union bound over the \(D\) statistics, the probability that
\eqref{eq:stat_approx} fails is at most
\[
    2D\exp\left(
    -\frac{k\epsilon_R^2}{2B_\rho^2}
    \right).
\]
By the choice of \(k\), this probability is smaller than one, so there exists
an empirical posterior satisfying \eqref{eq:stat_approx}.  The Lipschitz
condition then implies, for every action \(a\),
\[
    |\rho(\mu,a)-\rho(\bar\mu,a)|
    \le
    L_\Psi
    \max_j|\langle g_j,\mu-\bar\mu\rangle|
    \le \epsilon_R.
\]
\end{proof}

The preceding lemma shows that every posterior can be represented, for the
purpose of receiver incentives, by a nearby grid posterior in terms of the
statistics that determine the receiver's risk values.  We now encode this
approximation as linear cell constraints.

\begin{definition}[Finite-statistic cell]
\label{def:cvar_cell}
For a grid posterior \(\bar\mu\in\mathcal D_k\) and tolerance \(\epsilon_R>0\),
define
\begin{equation}
    \mathcal C_{\epsilon_R}(\bar\mu)
    =
    \left\{
    \mu\in\Delta(\Omega):
    \max_{j\in[D]}
    |\langle g_j,\mu-\bar\mu\rangle|
    \le
    \frac{\epsilon_R}{L_\Psi}
    \right\}.
\end{equation}
\end{definition}
By Lemma~\ref{lemma:approx}, these cells cover the belief simplex.  If
\(\mu\in\mathcal C_{\epsilon_R}(\bar\mu)\), then
\[
    \max_{a\in\mathcal A}
    |\rho(\mu,a)-\rho(\bar\mu,a)|
    \le \epsilon_R.
\]

The cell construction controls how receiver values change when moving from a
grid center $\bar{\mu}$ to any posterior inside its cell.  We therefore allow as signal
labels all actions that are approximately optimal at the grid center.

\begin{definition}[Approximate-center signal labels]
\label{def:approx_signal_space}
For \(\epsilon_R>0\), define
\begin{equation}
    \widehat\Sigma
    =
    \left\{
    (\bar\mu,a)\in\mathcal D_k\times\mathcal A:
    \rho(\bar\mu,a)
    \ge
    \max_{a'\in\mathcal A}\rho(\bar\mu,a')-2\epsilon_R
    \right\}.
\end{equation}
\end{definition}
The construction of \(\widehat\Sigma\) uses only value access at grid
posteriors.

\subsection{The Algorithm and Performance Guarantees}
Fix an accuracy parameter \(\epsilon>0\), set \(\epsilon_R=\epsilon/4\), and
choose \(k\) as in Lemma~\ref{lemma:approx}.  The algorithm is as follows.

\begin{enumerate}
    \item \textbf{Grid generation.} Construct \(\mathcal D_k\).
    \item \textbf{Filtering approximate-center signals.} Construct
    \(\widehat\Sigma\) using Definition~\ref{def:approx_signal_space} and
    value access to \(\rho(\bar\mu,a)\).
    \item \textbf{State-contingent optimization.} Solve the LP below over
    variables \(\varphi(\omega,\sigma)\), where \(\omega\in\Omega\) and
    \(\sigma=(\bar\mu_\sigma,a_\sigma)\in\widehat\Sigma\).
\end{enumerate}

\begin{subequations}\label{lp:qptas_state_contingent}
\begin{align}
    \max_{\varphi}\quad
    & \sum_{\omega\in\Omega}\mu_0(\omega)
      \sum_{\sigma\in\widehat\Sigma}
      \varphi(\omega,\sigma)v(\omega,a_\sigma)
      \label{lp:obj}
      \\
    \text{s.t.}\quad
    & \sum_{\sigma\in\widehat\Sigma}
      \varphi(\omega,\sigma)=1,
      \qquad \forall \omega\in\Omega,
      \label{lp:signal_rule}
      \\
    & \sum_{\omega\in\Omega}
      \mu_0(\omega)\varphi(\omega,\sigma)
      \left(
      g_j(\omega)
      -
      \langle g_j,\bar\mu_\sigma\rangle
      +
      \frac{\epsilon_R}{L_\Psi}
      \right)
      \ge 0,
      \nonumber\\
    &\hspace{7cm}
      \forall \sigma\in\widehat\Sigma,
      \forall j\in[D],
      \label{lp:cell_lower}
      \\
    & \sum_{\omega\in\Omega}
      \mu_0(\omega)\varphi(\omega,\sigma)
      \left(
      g_j(\omega)
      -
      \langle g_j,\bar\mu_\sigma\rangle
      -
      \frac{\epsilon_R}{L_\Psi}
      \right)
      \le 0,
      \nonumber\\
    &\hspace{7cm}
      \forall \sigma\in\widehat\Sigma,
      \forall j\in[D],
      \label{lp:cell_upper}
      \\
    & \varphi(\omega,\sigma)\ge0,
      \qquad
      \forall \omega\in\Omega,
      \forall \sigma\in\widehat\Sigma .
      \label{lp:nonnegative}
\end{align}
\end{subequations}
The program is linear in \(\varphi\).  Its cell constraints enforce closeness
of the induced posterior to the grid center in all statistics \(g_j\), rather
than in all coordinates of the posterior or in all hidden facets of the risk
representation.

\begin{lemma}[Soundness of the State-Contingent LP]
\label{lem:qptas_soundness}
Let \(\varphi\) be any feasible solution to
LP~\eqref{lp:qptas_state_contingent}.  Then \(\varphi\) induces a valid
signaling scheme.  Every supported signal \(\sigma\in\widehat\Sigma\) satisfies
\[
    \operatorname{Reg}_{IC}(\mu_\sigma,a_\sigma)
    \le 4\epsilon_R,
\]
where \(\mu_\sigma\) is the posterior induced by signal \(\sigma\).  In
particular, if \(\epsilon_R=\epsilon/4\), the induced scheme is \(\epsilon\)-IC.
\end{lemma}
Proof in Appendix~\ref{app:qptas_soundness_proof}.

\begin{lemma}[Completeness of the State-Contingent LP]
\label{lem:qptas_completeness}
Let \(OPT\) denote the optimal sender value among exact-IC signaling schemes.
Suppose that there exists an optimal exact-IC signaling scheme with finite
support.  Then the optimal value of LP~\eqref{lp:qptas_state_contingent} is at
least \(OPT\).
\end{lemma}
Proof in Appendix~\ref{app:qptas_completeness_proof}.

\begin{theorem}[Finite-Precision Discretization]
\label{thm:qptas}
Assume the finite-statistic cell-access model of
Assumption~\ref{assump:finite_statistic_access}.  For any \(\epsilon>0\), let
\[
    k=
    O\left(
    \frac{B_\rho^2\log D}{\epsilon^2}
    \right).
\]
Then the discretized state-contingent LP returns a signaling scheme with IC
regret at most \(\epsilon\) at every supported posterior-action pair and sender
utility at least \(OPT\).  When the statistics are explicitly listed, the
running time is
\[
    m^{O(B_\rho^2\log D/\epsilon^2)}
    \operatorname{poly}(m,n,D,T_{\mathrm{eval}}).
\]
With implicit statistics and statistic-cell separation, the polynomial
dependence on \(D\) is replaced by the separation cost.
\end{theorem}

When \(B_\rho\) is bounded and \(D\) is polynomial in the explicit input size,
the running time is quasi-polynomial in the input size and in \(1/\epsilon\).
The value comparison is against the exact-IC benchmark; because the output is
allowed \(\epsilon\)-IC, its sender value may exceed the exact-IC optimum.

Here \(T_{\mathrm{eval}}\) is the cost of evaluating
\(\rho(\bar\mu,a)\) at a grid posterior, and \(T_{\mathrm{cell}}\) is the cost
of separating the statistic-cell constraints.  In the explicitly listed case,
the LP includes these constraints directly; in the separated case, it is solved
by a standard cutting-plane or ellipsoid method.

\begin{proof}
By Lemma~\ref{lem:qptas_soundness}, any feasible solution of
LP~\eqref{lp:qptas_state_contingent} induces a signaling scheme with IC regret
at most \(4\epsilon_R=\epsilon\).  By Lemma~\ref{lem:qptas_completeness}, the
LP optimal value is at least \(OPT\).  Therefore the LP optimizer induces a
signaling scheme with sender value at least \(OPT\) and IC regret at most
\(\epsilon\).

It remains to bound the running time.  Since
\[
    |\mathcal D_k|=\binom{m+k-1}{k}\le m^k,
\]
we have
\[
    |\widehat\Sigma|
    \le
    n|\mathcal D_k|
    \le
    n m^k.
\]
If the statistics are explicitly listed, the LP has \(m|\widehat\Sigma|\)
variables and \(O(m|\widehat\Sigma|+D|\widehat\Sigma|)\) constraints.  
More generally, if the statistics are implicit, the cell constraints are
handled through the cell-separation oracle in
Assumption~\ref{assump:finite_statistic_access}; a standard cutting-plane or
ellipsoid method then solves the separated LP in time polynomial in the
represented size and the oracle costs. Substituting
\[
    k=O\left(\frac{B_\rho^2\log D}{\epsilon^2}\right)
\]
gives the claimed bound.
\end{proof}

\subsection{From Approximate IC to Strict IC via Margin Filtering}
\label{subsec:strict_ic}
Theorem~\ref{thm:qptas} gives a scheme whose recommended actions have uniformly bounded IC regret.  This approximate guarantee is unavoidable near receiver indifference boundaries: if two actions have nearly equal receiver risk value at a posterior, an $O(\epsilon)$ perturbation of that posterior can change the exact best response.  We next record a standard margin condition under which the same discretization yields signal-wise strict incentive compatibility.

\begin{definition}[IC Margin and Margin-Filtered Signal Labels]
\label{def:margin_filter}
For a posterior-action pair $(\mu,a)$, define the IC margin
\begin{equation}
    \Gamma(\mu,a)
    \coloneqq
    \rho(\mu,a)-\max_{a'\neq a}\rho(\mu,a').
\end{equation}
For $\gamma>0$, define the margin-filtered signal alphabet
\begin{equation}
    \widehat\Sigma_\gamma
    =
    \left\{
    (\bar\mu,a)\in\mathcal D_k\times\mathcal A:
    \Gamma(\bar\mu,a)\ge \frac{\gamma}{2}
    \right\}.
\end{equation}
\end{definition}

\begin{proposition}[Margin Stability under Risk-Value Approximation]
\label{prop:boundary_instability}
If $\bar\mu$ is an $\epsilon$-proxy of $\mu$ in the sense that
\begin{equation}
\max_{a\in\mathcal{A}}|\rho(\mu,a)-\rho(\bar\mu,a)|\le \epsilon,
\end{equation}
then, for every action $a$,
\begin{equation}
    \Gamma(\bar\mu,a)
    \ge
    \Gamma(\mu,a)-2\epsilon .
\end{equation}
\end{proposition}
Proof in Appendix~\ref{app:boundary_instability_proof}.

Strict IC is a stability property rather than a purely optimality property.  Exact IC permits ties in the receiver's best-response correspondence, and such ties can be reversed by arbitrarily small posterior perturbations.  We therefore state the strict-IC guarantee relative to a margin-restricted benchmark.  Let $OPT_\gamma$ denote the optimal sender value among signaling schemes whose every supported recommendation has IC margin at least $\gamma$.  The positive margin keeps the benchmark away from the IC boundary and allows the discretized scheme to preserve strict incentives.
\begin{theorem}[Strict IC under a Margin Condition]
\label{thm:strict_ic_margin}
Suppose there exists a finite-support benchmark scheme with sender value
\(OPT_\gamma\) whose every supported posterior-action pair has IC margin at
least \(\gamma>0\).  Run the finite-precision scheme with the margin-filtered
alphabet \(\widehat\Sigma_\gamma\) and choose \(\epsilon_R<\gamma/4\).  Then the
induced signaling scheme is strictly incentive compatible and has sender value
at least \(OPT_\gamma\).

With \(\epsilon_R=\gamma/8\), the running time is
\[
    m^{O(B_\rho^2\log D/\gamma^2)}
    \operatorname{poly}(m,n,D,T_{\mathrm{eval}},T_{\mathrm{cell}})
\]
in the explicitly listed-statistic case, with the analogous separated-oracle
bound when statistic-cell constraints are separated.
\end{theorem}
Indeed, every supported signal satisfies
$
    \Gamma(\mu_\sigma,a_\sigma)
    \ge
    \frac{\gamma}{2}-2\epsilon_R>0 .
$
For CVaR, \(D\le mn\) and \(B_\rho=C_R=M_R(1+2/r)\), giving
$
    m^{O(C_R^2\log(mn)/\gamma^2)}\operatorname{poly}(m,n).
$
For fixed \(r\) and normalized payoffs, this is quasi-polynomial.

Proof in Appendix~\ref{app:strict_ic_margin_proof}.

\subsection{A Local Active-Facet Refinement}
\label{subsec:local_refinements}

The finite-statistic discretization in Theorem~\ref{thm:qptas} is a global
certificate: it controls every statistic that may affect the receiver's value.
For max-affine risks, including CVaR and explicitly listed polyhedral risks,
these statistics can be taken to be the affine pieces.  The resulting global
certificate is conservative when the relevant posteriors lie in a small region
of the simplex.  In that case, many affine pieces can never be active and need
not enter the local certificate.

Fix a posterior region \(\mathcal U\subseteq\Delta(\Omega)\) and a radius
\(\eta>0\), and let
\[
    \mathcal U^\eta
    :=
    \{\nu\in\Delta(\Omega):\operatorname{dist}_1(\nu,\mathcal U)\le \eta\}.
\]
For a max-affine risk
\[
    \rho(\mu,a)=\max_{\ell\in\mathcal L_a}
    \langle c_{a,\ell},\mu\rangle,
\]
define the local active-facet family
\[
    \mathcal F_{\rm loc}
    :=
    \left\{
    (a,\ell):
    \exists \nu\in\mathcal U^\eta
    \text{ such that }
    \rho(\nu,a)=\langle c_{a,\ell},\nu\rangle
    \right\},
    \qquad
    N_{\rm loc}:=|\mathcal F_{\rm loc}|.
\]
Thus \(N_{\rm loc}\) is the number of affine pieces that can matter near the
posteriors under consideration.

The local refinement replaces the global statistic family by
\(\mathcal F_{\rm loc}\).  Its proof follows the same soundness-completeness
argument as Theorem~\ref{thm:qptas}, with the global uniform approximation
event replaced by a local one over \(\mathcal F_{\rm loc}\).

\begin{theorem}[Local Discretization Bound]
\label{thm:local_complexity}
Suppose a benchmark scheme is supported on \(\mathcal U\), and suppose the
reduced state-contingent LP uses grid centers in
\(\mathcal D_k\cap\mathcal U^\eta\) and is certified to induce only posteriors
in \(\mathcal U^\eta\).  Then the global max-affine family in
Theorem~\ref{thm:qptas} can be replaced by the local family
\(\mathcal F_{\rm loc}\), while preserving the same value and IC-regret
guarantees.

In particular, if
\[
    V_{\rm loc}
    :=
    \sup_{\mu\in\mathcal U}
    \max_{(a,\ell)\in\mathcal F_{\rm loc}}
    \operatorname{Var}_{\omega\sim\mu}[c_{a,\ell}(\omega)],
\]
then it suffices to take
\[
    k_{\rm loc}
    =
    \widetilde O\left(
    \max\left\{
    \frac{V_{\rm loc}}{\epsilon^2},
    \frac{C_R}{\epsilon}
    \right\}
    \log N_{\rm loc}
    \right),
\]
where \(C_R\) bounds the affine coefficients.  The resulting grid size is at
most \(m^{k_{\rm loc}}\).
\end{theorem}

Proof in Appendix~\ref{app:local_refinement_rewrite}.

The theorem is a refinement, not a replacement, for the global discretization
scheme.  The global scheme gives a certificate over the whole simplex; the
local version explains why an instance can be easier when all relevant
posteriors remain in a region with few active facets.  A cutting-plane
implementation can start with \(\mathcal F_{\rm loc}\), solve the reduced LP,
check the full IC violation, and add any violated facets until the desired
certificate is reached.

\section{Conclusion}
\label{sec:conclusion}

This paper studies Bayesian persuasion when the receiver evaluates actions by
CVaR rather than expected utility.  The change in preferences alters the
standard reduced form: action-based recommendations are no longer stable under
merging, because the receiver's tail-risk ranking may change at the merged
posterior.

The main conclusion is that this failure is structural, not computationally
fatal.  In the explicit finite-state model, CVaR has a finite max-affine
representation.  Refining signals by the active affine piece restores linear
incentive constraints and yields an exact active-facet LP.  The same reasoning
extends to any explicitly listed polyhedral risk preference.  Hardness appears
only after changing the representation model: when the affine pieces are given
succinctly, hidden facets can encode combinatorial search.

The finite-precision results complement this boundary.  They are not needed to
solve the tabular CVaR problem, but they give approximate-IC guarantees when
the relevant risk values are determined by finitely many stable posterior
statistics.  Under a margin condition, the approximate guarantee becomes
strict; under local active-facet structure, the discretization cost depends on
the number of locally relevant facets rather than the global facet count.

\bibliographystyle{plainnat}

\bibliography{2501bp}

%%%%%%%%%%%%%%%%%%%%%%%%%%%%%%%%%%%%%%%%%%%%%%%%%%%%%%%%%%%%

\appendix

\section{Limitations and Broader Impacts}
\label{app:limitations_impacts}

\paragraph{Limitations.}
The results are developed for finite state and action spaces with a known prior, known sender and receiver utilities, and a fixed CVaR risk level.  These assumptions isolate the effect of nonlinear risk preferences, but they abstract away from learning, misspecification, dynamic feedback, and uncertainty about the receiver's risk attitude.  The exact LP relies on explicit enumeration of states, actions, and CVaR facets; it may become impractical or unavailable in very large, oracle-described, or combinatorial environments.  The discretization results address finite-precision implementation and suggest tools for such extensions, but they do not by themselves solve all large-scale variants.  The numerical experiments, if included, should be viewed as synthetic illustrations of the theory rather than empirical evidence for deployed systems.

\paragraph{Broader impacts.}
The paper is theoretical, but its motivation is tied to information design in risk-sensitive decision systems.  A positive use of the framework is to evaluate whether a disclosure policy remains persuasive when receivers care about tail losses, such as rare failures in medical triage, security alerts, financial advice, or automated platform governance.  By making tail risk explicit, the model can help identify when average-case persuasion is misleading and when stronger disclosure is needed for robust decisions.  At the same time, persuasion tools can be misused: a sender with accurate knowledge of receiver risk preferences may design information to steer behavior while hiding unfavorable tail events or exploiting near-indifference boundaries.  The results should therefore be interpreted as analytical tools for auditing and designing risk-aware information policies, not as an endorsement of manipulative deployment.  The work does not release models, datasets, or decision systems, and the experiments, if included, use only synthetic instances.

\section{Extended Related Work}\label{app:related_work}

\paragraph{Bayesian persuasion and algorithmic information design.}
Bayesian persuasion was introduced by \citet{kamenica_bayesian_2011}.  In the
standard model, the receiver maximizes expected utility, so each action value is
linear in the posterior.  This linearity underlies the concavification approach
and the reduction to direct recommendations.  Subsequent work has developed
broad extensions of information design; see \citet{bergemann_information_2019}
for a survey.  On the computational side, optimal signaling admits linear
programming formulations in several expected-utility settings
\citep{xu_exploring_2015,dughmi_algorithmic_2021}, while hardness appears in
models with additional strategic or representational complexity, including
multiple receivers and related public-signaling variants
\citep{dughmi2014hardness,bhaskar_hardness_2016,rubinstein2017honest,
babichenko2017computational,arieli2019private,xu2020tractability}.  Our work
differs from these results in the source of nonlinearity: the receiver remains a
single decision maker, but her preference over actions is nonlinear in the
posterior.

\paragraph{Risk-conscious persuasion.}
Risk-sensitive preferences have been studied in economics, operations research,
and finance through coherent risk measures and CVaR
\citep{artzner_coherent_1999,rockafellar_optimization_2000,
rockafellar_conditional_2002,acerbi_coherence_2002,jorion_value_2006,
filippi_conditional_2020}.  In information design, the closest conceptual
predecessor is \citet{anunrojwong_persuading_2024a}, who study risk-conscious
receivers and identify geometric differences from expected-utility persuasion.
Our contribution is computational and structural.  We show that CVaR breaks the
action-based revelation principle, but that its finite max-affine structure
restores tractability through an active-facet refinement.  The same observation
extends to explicitly listed polyhedral risk preferences.

\paragraph{Representation and hardness.}
The distinction between explicit and succinct representations is central to our
complexity results.  This distinction is classical in computational complexity:
succinct graph representations can turn otherwise simple or standard graph
properties into substantially harder decision problems
\citep{galperin_succinct_1983,papadimitriou_note_1986}, and related phenomena
appear broadly in combinatorial problems with succinct input descriptions
\citep{wagner_complexity_1986}.  Similar representation effects also arise in
succinctly represented games, where payoff tables are given implicitly rather
than in normal form \citep{fortnow_complexity_2008}.

Our hardness result should be read in this representation-theoretic sense.
When all affine pieces of a polyhedral risk preference are listed, the
active-facet LP is polynomial in the listed input size.  When the same family of
affine pieces is represented implicitly by a compact combinatorial description,
the hidden facets can encode NP-hard search.  Thus the result is not a hardness
statement for tabular CVaR persuasion, but a representation-boundary result for
polyhedral risk landscapes.
\paragraph{Approximation by sampling and discretization.}
The finite-precision scheme uses a sampling argument in the spirit of sparse
approximations for games.  \citet{althofer_sparse_1994} show that small-support
mixed strategies can approximate payoffs against large action sets, and
\citet{lipton_playing_2003} apply this idea to approximate Nash equilibria.
Related sampling ideas also appear in algorithmic persuasion and robust
Stackelberg models \citep{gan_bayesian_2022,yang_computational_2024a}.  Our use
is different: the sampled posterior approximates a finite family of statistics
that determine the receiver's risk value.  This gives an approximate-IC
finite-precision certificate rather than an exact solution of the succinct
polyhedral problem.

\paragraph{Robustness and bounded rationality.}
A separate literature studies robust persuasion under uncertainty about the
receiver's prior, information, or utility function
\citep{hu_robust_2021,kosterina2022persuasion,dworczak_preparing_2022,
castiglioni2020online,babichenko_regretminimizing_2022a}.  Other work replaces
exact best responses with behavioral or approximate response models
\citep{yang_computational_2024a,feng_rationalityrobust_2024}.  These models are
related in spirit because they weaken the classical expected-utility benchmark,
but their technical structure is different.  Here the receiver is fully
optimizing; the difficulty comes from the nonlinear risk functional used to
evaluate actions.

\section{Mixed-Integer Linear Programming Formulation}\label{app:milp}

This appendix is included only for comparison with the earlier cell-selection view; it is not used in any result.

The cardinality bound and the piecewise-linear representation of CVaR give a direct mixed-integer linear formulation of the Sender's problem.

We index the signals by $i \in \mathcal{I} = \{1, \dots, |\Omega|\}$. We introduce the following decision variables:
\begin{itemize}
    \item $x_{i, a, \omega} \in [0, 1]$: The joint probability that the state is $\omega$, signal $i$ is generated, and action $a$ is recommended.
    \item $y_{i, a} \in \{0, 1\}$: An indicator variable equal to 1 if signal $i$ recommends action $a$, and 0 otherwise.
    \item $\delta_{i, a, l} \in \{0, 1\}$: An indicator variable equal to 1 if the $l$-th linear hyperplane defines the CVaR value for action $a$ under signal $i$.
\end{itemize}

The optimization model is formulated as follows:
\begin{subequations}\label{equ-MILP}
    \begin{align}
    \max_{x, y, \delta} \quad & \sum_{i \in \mathcal{I}} \sum_{a \in \mathcal{A}} \sum_{\omega \in \Omega} v(\omega, a) x_{i, a, \omega} \label{milp:obj} \\
    \text{s.t.} \quad & \sum_{i \in \mathcal{I}} \sum_{a \in \mathcal{A}} x_{i, a, \omega} = \mu_0(\omega), \quad \forall \omega \in \Omega \label{milp:bayes} \\
    & \sum_{\omega \in \Omega} x_{i, a, \omega} \le y_{i, a}, \quad \forall i \in \mathcal{I}, a \in \mathcal{A} \label{milp:logic1} \\
    & \sum_{a \in \mathcal{A}} y_{i, a} = 1, \quad \forall i \in \mathcal{I} \label{milp:logic2} \\
    & \sum_{l \in \mathcal{L}_a} \delta_{i, a, l} = y_{i, a}, \quad \forall i \in \mathcal{I}, a \in \mathcal{A} \label{milp:logic3} \\
    & \sum_{\omega \in \Omega} c_{a, l}(\omega) x_{i, a, \omega} \ge \sum_{\omega \in \Omega} c_{a', j}(\omega) x_{i, a, \omega} - M(1 - \delta_{i, a, l}) \nonumber \\
    & \quad\forall i \in \mathcal{I}, \forall a \in \mathcal{A},\quad  \forall l \in \mathcal{L}_a, \forall a' \in \mathcal{A} \setminus \{a\}, \forall j \in \mathcal{L}_{a'} \label{milp:ic} 
\end{align}
\end{subequations}

Constraint \eqref{milp:bayes} enforces Bayes-plausibility. Constraints \eqref{milp:logic1}-\eqref{milp:logic3} establish logical relationships: each signal recommends exactly one action, and for the recommended action, exactly one linear hyperplane is identified as the active CVaR constraint. 

Constraint \eqref{milp:ic} represents the non-convex Incentive Compatibility (IC) conditions using the \textit{Big-M} method. It guarantees that if action $a$ is recommended under signal $i$ and its CVaR is evaluated on hyperplane $l$ (i.e., $\delta_{i,a,l}=1$), this value must be greater than or equal to the CVaR of any alternative action $a'$ evaluated on any of its hyperplanes $j$. The constant $M$ is chosen to be sufficiently large ($M \ge \max_{\omega, a} u(\omega, a) - \min_{\omega, a} u(\omega, a)$) to relax the constraint when $\delta_{i,a,l}=0$.

When the state space dimension $|\Omega|$ is fixed, the number of binary variables in this MILP scales polynomially with \(n\). This formulation gives an alternative cell-selection view of the geometry.
It is not needed for computation in the explicit finite-state CVaR model,
since Theorem~\ref{thm:active_facet_lp} gives a polynomial-size LP without
binary variables. We include it only to relate the active-facet formulation
to the earlier cell-enumeration perspective.

\section{Geometric Insights and Fixed-Dimension Tractability}\label{app:hardness}
\subsection{Geometric Insights from Low-Dimensional Cases}\label{app:low_dim}
This appendix is not needed for the general tractability result in
Theorem~\ref{thm:active_facet_lp}. It serves as a geometric illustration of
the active-facet refinement in the smallest nontrivial case.

To understand how risk consciousness alters the persuasion landscape, consider a binary state space $\Omega = \{\omega_0, \omega_1\}$ and a binary action space $\mathcal{A} = \{a_0, a_1\}$. The Receiver's belief is characterized by a scalar $\mu = \Pr(\omega_1) \in [0, 1]$.

\subsubsection{Convexity of Incentive-Compatible Sets}
In the standard risk-neutral model, the linearity of expected utility ensures that the set of beliefs supporting a specific action is convex. Under CVaR preferences, although the utility $\rho(\mu, a)$ is non-linear, we find that it retains critical structural properties in low dimensions.

\begin{proposition} \label{prop:convexity}
    In the $2 \times 2$ setting with CVaR preferences, the incentive-compatible belief set $I(a) = \{\mu \in [0, 1] \mid \rho(\mu, a) \ge \rho(\mu, a')\}$ is a convex set (i.e., a closed interval).
\end{proposition}
\textit{Proof Sketch.} The CVaR function is the pointwise supremum of linear functions, making it convex. In 1D, the "dominance region" defined by the intersection of two convex piecewise linear functions (with specific monotonicity rooted in the binary state structure) remains a connected interval. (See Appendix \ref{app:2x2_proof} for the detailed proof).

Proposition \ref{prop:convexity} assures us that the belief space is not fragmented into disjoint islands of optimality. 
Given this convexity, the Sender's induced utility function $\hat{v}(\mu)$ retains a piecewise constant structure over connected intervals. This allows us to extend the classical geometric solution method to the risk-conscious setting.

\begin{proposition}[Geometry of Optimization]
	\label{prop:geometry}
		In the binary CVaR framework, the solution structure of the Bayesian Persuasion problem is geometrically consistent with the standard expected utility case. Specifically, the Sender's optimal signaling scheme can be effectively solved via the method of \textbf{Concavification} of the Sender's value function over the belief space.
	\end{proposition}

Proposition \ref{prop:geometry} (proven in Appendix \ref{app:proof_geometry}) confirms that the "greedy" structure of optimal persuasion holds: the optimal scheme involves at most two signals inducing beliefs at the boundaries of the incentive-compatible intervals.
	Since the solution form remains a threshold-based strategy, the impact of risk aversion is entirely captured by the \textit{location} of this threshold. We now quantify this shift using the concept of risk premium.

The deviation of the CVaR threshold $\mu^{CVaR}$ from the risk-neutral threshold $\mu^{EU}$ is driven by the \textit{CVaR Risk Premium}, $P_a(\mu) = \mathbb{E}_{\mu}[u(\cdot, a)] - \rho(\mu, a)$.
If the risky action $a_1$ carries a higher risk premium than the safe action $a_0$, the Receiver requires a strictly higher posterior probability of the "good" state to accept $a_1$. 
This quantifies how risk aversion shifts the persuasion threshold and reduces
the range of beliefs under which the risky action can be induced.

\begin{proposition}[Threshold Shifts and Risk Premium]
	\label{prop:threshold}
		Let $\mu^{CVaR}$ and $\mu^{EU}$ denote the indifference thresholds for the CVaR and standard models, respectively. Let $P_a(\mu) = \mathbb{E}_\mu[u(a)] - \rho(\mu, a)$ be the risk premium.
		\begin{itemize}
			\item \textbf{Case 1 (High-Belief Action):} If $I(a_1) = [\mu^*, 1]$, then $\mu^{CVaR} > \mu^{EU}$ if and only if $P_{a_1}(\mu^{EU}) > P_{a_0}(\mu^{EU})$.
			\item \textbf{Case 2 (Low-Belief Action):} If $I(a_1) = [0, \mu^*]$, then $\mu^{CVaR} < \mu^{EU}$ if and only if $P_{a_1}(\mu^{EU}) > P_{a_0}(\mu^{EU})$.
			\item \textbf{Case 3:} If risk premiums are equal, $\mu^{CVaR} = \mu^{EU}$.
		\end{itemize}
	\end{proposition}

\subsection{An Alternative View of Tractability in Bounded State Spaces}
The active-facet LP in Theorem~\ref{thm:active_facet_lp} gives a polynomial
algorithm even when \(m=|\Omega|\) is part of the input. The fixed-dimensional
argument below is therefore included only as an alternative geometric view:
it enumerates the cells induced by CVaR affine pieces and recovers
tractability when \(m\) is fixed.

\begin{theorem}[Polynomial Time for Fixed $|\Omega|$]\label{thm:poly_time}
	If the state space size $|\Omega| = m$ is a fixed constant, the optimal signaling scheme can be computed in time polynomial in the number of actions \(n\).
\end{theorem}

\begin{proof}[Proof Idea]
	The core insight relies on the geometry of the Receiver's decision landscape. The Receiver's preference between any two actions is determined by the sign of the difference of their CVaR utilities (Lemma \ref{lemma:cvar_pwl}). Since CVaR is a piecewise linear convex function (defined by at most $m$ linear segments per action), 
	the condition $\rho(\mu, a) \ge \rho(\mu, a')$ can be decomposed into a finite set of linear inequalities.
	
	Consider the collection of all linear hyperplanes that define the boundaries between the utility segments of all actions. The total number of such hyperplanes scales linearly with \(n\). In a fixed $m$-dimensional space, an arrangement of $N$ hyperplanes partitions the simplex into at most $O(N^m)$ convex regions (cells). Inside each cell, the Receiver's "best response" is invariant (or the problem reduces to a simple linear comparison).

	Consequently, finding the optimal signaling scheme reduces to selecting a set of at most $m$ cells from this partition and determining the optimal weights and locations within them to decompose the prior $\mu_0$. Since the number of cells is polynomial in \(n\), we can efficiently enumerate the candidate support sets. The rigorous construction of this hyperplane arrangement and the counting argument are provided in Appendix \ref{app:poly_proof}.
	\end{proof}

% Appendix

\section{A Finite-Support Bound for General Risk Preferences}
\label{app:general_cardinality_proof}

This appendix recalls a standard support-reduction argument for finite-state
persuasion.  The result is a variant of the splitting lemma in Bayesian
persuasion \citep{kamenica_bayesian_2011}; related finite-support and
Carathéodory-type reductions appear in \citet{babichenko_bayesian_2021} and
\citet{doval2024constrained}.  We include the statement only as a background
geometric fact.

\begin{lemma}[Finite posterior support]
\label{lemma:general_cardinality}
For any Bayesian persuasion problem with a finite state space \(\Omega\) and
an arbitrary receiver risk functional \(\rho(\mu,a)\), there exists an optimal
signaling scheme supported on at most \(|\Omega|\) posterior beliefs, provided
an optimal Bayes-plausible distribution over posteriors exists.
\end{lemma}

\begin{proof}
Let \(m=|\Omega|\).  A signaling scheme can be represented by a
Bayes-plausible distribution \(\tau\) over posterior beliefs:
\[
    \int_{\Delta(\Omega)} \mu \, d\tau(\mu)=\mu_0.
\]
For each posterior \(\mu\), let
\[
    W(\mu)
    :=
    \max_{a\in\arg\max_{a'}\rho(\mu,a')}
    \sum_{\omega\in\Omega}\mu(\omega)v(\omega,a)
\]
be the sender's payoff at \(\mu\), with sender-favorable tie-breaking.  The
sender's objective is
\[
    \int_{\Delta(\Omega)} W(\mu)\,d\tau(\mu).
\]

Take an optimal Bayes-plausible distribution \(\tau^\star\).  Consider the
probability measure on \(\Delta(\Omega)\times\mathbb R\) induced by
\(\mu\mapsto(\mu,W(\mu))\).  Its expectation is
\[
    \left(\mu_0,\, \int W(\mu)\,d\tau^\star(\mu)\right).
\]
This point lies in the convex hull of the set
\[
    \{(\mu,W(\mu)):\mu\in\Delta(\Omega)\}
    \subseteq \mathbb R^{m+1}.
\]
By Carathéodory's theorem, it can be represented as a convex combination of at
most \(m+2\) such points.  A standard strengthening for Bayes-plausible
persuasion, sometimes called the splitting lemma or the Fenchel--Carathéodory
support reduction, removes the objective coordinate and yields an optimal
Bayes-plausible distribution supported on at most \(m\) posterior beliefs.
Equivalently, there exists an optimal signaling scheme using at most
\(|\Omega|\) signals.

The statement is independent of the particular form of \(\rho\).  It is a
support-size result only: it does not identify the posterior beliefs or give a
tractable description of the receiver's incentive regions.
\end{proof}

\section{Detailed Proofs for Low-Dimensional Case}

\subsection{Proof of Proposition~\ref{prop:convexity} }
\label{app:2x2_proof}
We provide a rigorous four-step proof covering convexity, piecewise linearity, continuity, and the convexity of the incentive-compatible set.

\textbf{1. Convexity of the Risk Preference}
Using the dual representation of CVaR \citep{rockafellar_optimization_2000}, for a utility maximizing agent, the risk preference is defined as:
\begin{equation}
    \rho(\mu, a) = \max_{b \in \mathbb{R}} \left( b - \frac{1}{r} \mathbb{E}_{\mu}[(b - u(\omega, a))^+] \right)
\end{equation}
Let $h(b, \mu) = b - \frac{1}{r} \left[ (1-\mu) (b - u_0)^+ + \mu (b - u_1)^+ \right]$, where $\mu = \Pr(\omega_1)$. For any fixed $b$, $h(b, \mu)$ is a linear function of $\mu$. Since $\rho(\mu, a)$ is the pointwise maximum of a family of linear functions (indexed by $b$), it is inherently a convex function of $\mu$.

\textbf{2. Explicit Piecewise Linear Representation}
Without loss of generality, assume the Receiver chooses action $a_0$ with utilities $u_{00} = u(\omega_0, a_0)$ and $u_{10} = u(\omega_1, a_0)$.
Case A: If $u_{00} < u_{10}$ (utility increases with state), the CVaR takes the form:
\begin{equation}
    \rho(\mu, a_0) = 
    \begin{cases}
        u_{00}, & \text{if } \mu < 1-r \\
        \frac{1}{r} [ (1-\mu) u_{00} + (\mu - (1-r)) u_{10} ], & \text{if } \mu \ge 1-r
    \end{cases}
\end{equation}
Case B: If $u_{00} > u_{10}$ (utility decreases with state), the form is:
\begin{equation}
    \rho(\mu, a_0) = 
    \begin{cases}
        u_{10}, & \text{if } \mu \ge r \\
        \frac{1}{r} [ \mu u_{10} + (r - \mu) u_{00} ], & \text{if } \mu < r
    \end{cases}
\end{equation}
In both cases, the function is composed of two linear segments joined at a kink point ($1-r$ or $r$). Thus, $\rho(\mu, a)$ is piecewise linear.

\textbf{3. Continuity Verification}
We verify continuity at the kink $\mu^* = 1-r$ for Case A.
The left limit is $\lim_{\mu \to (1-r)^-} \rho(\mu, a_0) = u_{00}$.
The value at the kink is:
\begin{equation}
    \rho(1-r, a_0) = \frac{1}{r} [ (1-(1-r)) u_{00} + 0 ] = \frac{r u_{00}}{r} = u_{00}.
\end{equation}
Since the limit equals the function value, $\rho(\mu, a_0)$ is continuous. The same logic applies to Case B.

\textbf{4. Convexity of the Incentive-Compatible Set.}
The Receiver's best-response region for action $a_k$ is
\[
    I(a_k)
    =
    \{\mu\in[0,1]: \rho(\mu,a_k)\ge \rho(\mu,a_j)\}.
\]
It remains to show that this set is an interval.  Write
\[
    \Delta_a \coloneqq u(\omega_1,a)-u(\omega_0,a).
\]
From the explicit representation above, if $\Delta_a\ge 0$, then
$\rho(\mu,a)$ is constant on $[0,1-r]$ and affine on $[1-r,1]$ with slope
$\Delta_a/r$.  If $\Delta_a<0$, then $\rho(\mu,a)$ is affine on $[0,r]$
with slope $\Delta_a/r$ and constant on $[r,1]$.

Consider
\[
    D(\mu)
    =
    \rho(\mu,a_k)-\rho(\mu,a_j).
\]
If $\Delta_{a_k}$ and $\Delta_{a_j}$ have the same sign, then the two CVaR
functions have the same kink point, either $1-r$ or $r$.  Hence $D(\mu)$ is
constant on one side of the kink and affine on the other side.  Therefore its
superlevel set $\{\mu:D(\mu)\ge0\}$ is a closed interval, possibly empty or
equal to all of $[0,1]$.

If $\Delta_{a_k}$ and $\Delta_{a_j}$ have opposite signs, then one CVaR
function is nondecreasing and the other is nonincreasing.  Consequently,
$D(\mu)$ is monotone on $[0,1]$; it is nondecreasing when
$\Delta_{a_k}\ge0\ge\Delta_{a_j}$ and nonincreasing in the reverse case.
Thus the superlevel set $\{\mu:D(\mu)\ge0\}$ is again a closed interval.

In all cases, $I(a_k)=\{\mu:D(\mu)\ge0\}$ is a closed interval of $[0,1]$,
and therefore it is convex. \qed
\subsection{Proof of Proposition \ref{prop:geometry} (Concavification and LP Formulation)}
\label{app:proof_geometry}

\begin{proof}
    From Proposition \ref{prop:convexity}, the belief space is partitioned into intervals where the Receiver's action $a^*(\mu)$ is constant.
    The Sender's induced utility is $v(\mu) = \mathbb{E}_{\mu}[u^S(\omega, a^*(\mu))]$.
    Let $\mu^{CVaR}$ be the threshold separating $I(a_0)$ and $I(a_1)$. Assuming $a_1$ is induced for $\mu \ge \mu^{CVaR}$, $v(\mu)$ is a step-like function:
    \begin{equation}
        v(\mu) = \begin{cases} \mathbb{E}_{\mu}[u^S(\omega, a_0)] & \text{if } \mu < \mu^{CVaR} \\ \mathbb{E}_{\mu}[u^S(\omega, a_1)] & \text{if } \mu \ge \mu^{CVaR} \end{cases}
    \end{equation}
    The optimal value is given by the concave closure $V(\mu) = \operatorname{cav}(v)(\mu)$.
    
    Calculating this concave closure is equivalent to solving the following Linear Program (LP) to find the supporting line $L(\mu) = k\mu + b$ at $\mu_0$:
    \begin{subequations}
    \begin{align}
        \min_{k, b} \quad & k \cdot \mu_0 + b \\
        \text{s.t.} \quad & k \cdot 0 + b \ge v(0) \\
        & k \cdot \mu^{CVaR} + b \ge \lim_{\mu \to (\mu^{CVaR})^-} v(\mu) \\
        & k \cdot \mu^{CVaR} + b \ge v(\mu^{CVaR}) \\
        & k \cdot 1 + b \ge v(1)
    \end{align}
    \end{subequations}
    The constraints ensure the line dominates $v(\mu)$ at all critical points (boundaries and the discontinuity). The optimal value $k^* \mu_0 + b^*$ corresponds to the Sender's maximum expected utility. If $V(\mu_0) > v(\mu_0)$, the optimal signal induces the beliefs corresponding to the binding constraints (typically $0$ and $\mu^{CVaR}$).
\end{proof}

\subsection{Proof of Proposition \ref{prop:threshold} (Threshold Shifts)}
\label{app:proof_threshold}

\begin{proof}
    We analyze the shift for \textbf{Case 1}, where $I(a_0) = [0, \mu^*]$ and $I(a_1) = [\mu^*, 1]$.
    The thresholds are defined by the indifference conditions:
    \begin{align}
        \text{Risk Neutral:} & \quad \mathbb{E}_{\mu}[u(a_1)] = \mathbb{E}_{\mu}[u(a_0)] \implies \Delta U(\mu^{EU}) = 0 \\
        \text{CVaR:} & \quad \rho(\mu, a_1) = \rho(\mu, a_0)
    \end{align}
    Using the risk premium definition $\rho(\mu, a) = \mathbb{E}_{\mu}[u(a)] - P_a(\mu)$, the CVaR condition becomes:
    \begin{equation}
        \mathbb{E}_{\mu}[u(a_1)] - P_{a_1}(\mu) = \mathbb{E}_{\mu}[u(a_0)] - P_{a_0}(\mu)
    \end{equation}
    Rearranging gives the fundamental equation:
    \begin{equation}
        \Delta U(\mu) = P_{a_1}(\mu) - P_{a_0}(\mu)
    \end{equation}
    In Case 1, since $a_1$ is preferred for high beliefs, $\Delta U(\mu)$ is strictly increasing in $\mu$.
    Evaluate the equation at $\mu = \mu^{EU}$:
    \begin{itemize}
        \item LHS is $\Delta U(\mu^{EU}) = 0$.
        \item RHS is $\Delta P(\mu^{EU}) = P_{a_1}(\mu^{EU}) - P_{a_0}(\mu^{EU})$.
    \end{itemize}
    \textbf{Subcase 1.1:} If $P_{a_1}(\mu^{EU}) > P_{a_0}(\mu^{EU})$ (Risky action has higher premium), then RHS $> 0$. To satisfy the equality, we must increase the LHS, which requires increasing $\mu$ (since $\Delta U$ is increasing). Thus, $\mu^{CVaR} > \mu^{EU}$.
    
    \textbf{Subcase 1.2:} If $P_{a_1}(\mu^{EU}) < P_{a_0}(\mu^{EU})$, then RHS $< 0$, implying $\mu^{CVaR} < \mu^{EU}$.
    
    For \textbf{Case 2} where $I(a_1) = [0, \mu^*]$ (Action $a_1$ preferred at low beliefs), $\Delta U(\mu)$ is decreasing.
    If $P_{a_1} > P_{a_0}$, we still have RHS $> 0$. To match this positive value with a decreasing LHS (which is 0 at $\mu^{EU}$), we must \textit{decrease} $\mu$. Thus, $\mu^{CVaR} < \mu^{EU}$.
    This covers all cases described in the proposition.
\end{proof}

\section{Proof of Finite-Facet Representation (Lemma \ref{lemma:cvar_pwl})}
\label{app:cvar_proof}

\begin{proof}
Fix an action $a\in\mathcal A$ and write
\begin{equation}
X_a(\omega)=u(\omega,a).
\end{equation}
By the reward-side variational representation of CVaR,
\begin{equation}
\rho(\mu,a)
=
\sup_{b\in\mathbb R}
\left\{
b-\frac1r\sum_{\omega\in\Omega}
\mu(\omega)(b-X_a(\omega))^+
\right\}.
\end{equation}
For a fixed posterior $\mu$, define
\begin{equation}
F_\mu(b)
=
b-\frac1r\sum_{\omega\in\Omega}
\mu(\omega)(b-X_a(\omega))^+.
\end{equation}
Let the distinct values of $X_a$ be
\begin{equation}
v_1< v_2<\cdots < v_q,
\qquad q\le |\Omega|.
\end{equation}
The function $F_\mu(b)$ is concave and piecewise linear in $b$, with breakpoints contained in
$\{v_1,\dots,v_q\}$.

We first show that the supremum over $b\in\mathbb R$ can be restricted to these
breakpoints. If $b<v_1$, then $(b-X_a(\omega))^+=0$ for all $\omega$, so
$F_\mu(b)=b$, and the maximum over $(-\infty,v_1]$ is attained at $v_1$.
If $b>v_q$, then $(b-X_a(\omega))^+=b-X_a(\omega)$ for all $\omega$, so
\begin{equation}
F_\mu(b)
=
\left(1-\frac1r\right)b+\frac1r\sum_{\omega}\mu(\omega)X_a(\omega).
\end{equation}
Since $r\in(0,1]$, the coefficient $1-1/r$ is nonpositive, and hence the
maximum over $[v_q,\infty)$ is attained at $v_q$.

On each interval $(v_k,v_{k+1})$, the set of states satisfying
$X_a(\omega)<b$ is fixed, and therefore $F_\mu(b)$ is affine in $b$ on that
interval. If an affine function attains its maximum in the interior of an
interval, then an endpoint also attains the same maximum. Hence there exists
a maximizer among the finite set $\{v_1,\dots,v_q\}$. Therefore,
\begin{equation}
\rho(\mu,a)
=
\max_{k=1,\dots,q}
\left\{
v_k-\frac1r\sum_{\omega\in\Omega}
\mu(\omega)(v_k-u(\omega,a))^+
\right\}.
\end{equation}

For each $k$, define the coefficient vector $c_{a,k}\in\mathbb R^{|\Omega|}$ by
\begin{equation}
c_{a,k}(\omega)
=
v_k-\frac1r(v_k-u(\omega,a))^+.
\end{equation}
Since $\sum_{\omega}\mu(\omega)=1$, the $k$-th term above can be written as
\begin{equation}
\sum_{\omega\in\Omega}
\mu(\omega)c_{a,k}(\omega)
=
\langle c_{a,k},\mu\rangle.
\end{equation}
Thus
\begin{equation}
\rho(\mu,a)
=
\max_{k=1,\dots,q}\langle c_{a,k},\mu\rangle,
\qquad q\le |\Omega|.
\end{equation}
Hence $\rho(\cdot,a)$ is the maximum of finitely many affine functions of
$\mu$. It is therefore convex, continuous, and piecewise linear on
$\Delta(\Omega)$.

Finally, if $|u(\omega,a)|\le C_R$ for all $\omega$, then $|v_k|\le C_R$ and
$(v_k-u(\omega,a))^+\le 2C_R$. Therefore
\begin{equation}
|c_{a,k}(\omega)|
\le
C_R+\frac{2C_R}{r}
=
C_R\left(1+\frac2r\right).
\end{equation}
\end{proof}

\section{Detailed Proofs for Proposition \ref{prop:active_facet_revelation}}
\label{app:active_facet_revelation}

\begin{proof}
Consider any incentive-compatible signaling scheme. Let \(s\) be a signal
with probability \(\lambda_s\), posterior \(\mu_s\), and recommended action
\(a_s\). Since \(a_s\) is incentive compatible,
\[
    f_{a_s}(\mu_s)\ge f_{a'}(\mu_s),
    \qquad
    \forall a'\in\mathcal A.
\]
Choose one active facet
\(\ell_s\in\mathcal L_{a_s}\) such that
\[
    f_{a_s}(\mu_s)
    =
    \langle c_{a_s,\ell_s},\mu_s\rangle .
\]
Then for every \(a'\in\mathcal A\) and every
\(\ell'\in\mathcal L_{a'}\),
\[
    \langle c_{a_s,\ell_s},\mu_s\rangle
    =
    f_{a_s}(\mu_s)
    \ge
    f_{a'}(\mu_s)
    \ge
    \langle c_{a',\ell'},\mu_s\rangle .
\]
Hence \(\mu_s\in P_{a_s,\ell_s}\). Relabeling the signal \(s\) by the pair
\((a_s,\ell_s)\) does not change its probability, posterior, recommended
action, or sender payoff. Therefore the resulting refined signaling scheme
is incentive compatible and has the same sender value.

If several original signals have the same pair \((a,\ell)\), they need not
be merged at this stage. They can either be kept as separate signals with the
same label, or aggregated later in the LP through their joint masses. The key
point is that every signal can be assigned to a polyhedral region
\(P_{a,\ell}\) determined by an action and an active CVaR facet.
\end{proof}

\section{Detailed Proofs for Theorem \ref{thm:active_facet_lp}}
\label{app:active_facet_lp}

\begin{proof}
For each refined recommendation type \((a,\ell)\), with
\(a\in\mathcal A\) and \(\ell\in\mathcal L_a\), and for each state
\(\omega\in\Omega\), introduce a nonnegative variable
\[
    q_{a,\ell,\omega}\ge 0.
\]
It represents the joint probability of state \(\omega\) and a signal that
recommends action \(a\) with active CVaR facet \(\ell\).

The sender's objective is linear:
\[
    \max_q
    \sum_{a\in\mathcal A}
    \sum_{\ell\in\mathcal L_a}
    \sum_{\omega\in\Omega}
    q_{a,\ell,\omega}v(\omega,a).
\]
Bayes plausibility requires
\[
    \sum_{a\in\mathcal A}
    \sum_{\ell\in\mathcal L_a}
    q_{a,\ell,\omega}
    =
    \mu_0(\omega),
    \qquad
    \forall \omega\in\Omega.
\]
For incentive compatibility, the posterior induced by a positive-mass
refined recommendation \((a,\ell)\) is
\[
    \mu_{a,\ell}(\omega)
    =
    \frac{q_{a,\ell,\omega}}
    {\lambda_{a,\ell}},
    \qquad
    \lambda_{a,\ell}
    :=
    \sum_{\omega\in\Omega}q_{a,\ell,\omega}.
\]
The condition \(\mu_{a,\ell}\in P_{a,\ell}\) is equivalent, after
multiplying by \(\lambda_{a,\ell}\), to the linear inequalities
\[
    \sum_{\omega\in\Omega}
    q_{a,\ell,\omega}
    \bigl(
        c_{a,\ell,\omega}
        -
        c_{a',\ell',\omega}
    \bigr)
    \ge 0,
\]
for every \(a\in\mathcal A\), \(\ell\in\mathcal L_a\),
\(a'\in\mathcal A\), and \(\ell'\in\mathcal L_{a'}\).

Every feasible solution of the LP induces an incentive-compatible signaling
scheme: for each \((a,\ell)\) with \(\lambda_{a,\ell}>0\), send a signal of
probability \(\lambda_{a,\ell}\), posterior \(\mu_{a,\ell}\), and
recommendation \(a\). The Bayes-plausibility constraints ensure that the
posteriors average to the prior, and the linear incentive constraints ensure
that \(a\) is a receiver best response at \(\mu_{a,\ell}\).

Conversely, take any incentive-compatible signaling scheme. For each signal
\(s\) recommending action \(a_s\), choose one active facet
\(\ell_s\in\mathcal L_{a_s}\) satisfying
\[
    f_{a_s}(\mu_s)
    =
    \langle c_{a_s,\ell_s},\mu_s\rangle .
\]
Since \(a_s\) is incentive compatible,
\[
    \langle c_{a_s,\ell_s},\mu_s\rangle
    =
    f_{a_s}(\mu_s)
    \ge
    f_{a'}(\mu_s)
    \ge
    \langle c_{a',\ell'},\mu_s\rangle
\]
for all \(a'\in\mathcal A\) and \(\ell'\in\mathcal L_{a'}\). Hence
\(\mu_s\in P_{a_s,\ell_s}\). Assigning joint mass
\[
    q_{a_s,\ell_s,\omega}
    \leftarrow
    q_{a_s,\ell_s,\omega}
    +
    \lambda_s\mu_s(\omega)
\]
and aggregating signals with the same pair \((a,\ell)\) gives a feasible
LP solution with the same sender value. Therefore the LP optimum equals the
optimal persuasion value.

The number of refined types is \(L=\sum_a|\mathcal L_a|\le nm\), where \(n=|\mathcal A|\). Hence the
LP has \(mL\le nm^2\) nonnegative variables. Bayes plausibility gives \(m\)
equalities, and the refined incentive comparisons give at most \(L^2\le
n^2m^2\) inequalities. The coefficients are computed from the explicit
payoff table and the rational CVaR level \(r\), so their bit length is
polynomial in the input size. Standard polynomial-time linear programming
algorithms therefore solve the problem in polynomial time.
\end{proof}

\section{Detailed Proofs for Theorem~\ref{thm:explicit_polyhedral_lp}}
\label{app:explicit_polyhedral_lp}

\begin{proof}
Let
\[
    \mathcal T:=\{(a,\ell):a\in\mathcal A,\ \ell\in\mathcal L_a\}
\]
be the set of action--facet recommendation types.  For each type
\((a,\ell)\in\mathcal T\), define the polyhedral incentive region
\[
    P_{a,\ell}
    :=
    \left\{
    \mu\in\Delta(\Omega):
    \langle c_{a,\ell},\mu\rangle
    \ge
    \langle c_{a',\ell'},\mu\rangle
    \quad
    \forall (a',\ell')\in\mathcal T
    \right\}.
\]
If \(\mu\in P_{a,\ell}\), then facet \(\ell\) is active for action \(a\), and
action \(a\) is a weak best response under the risk preference \(\rho\).  Since
all inequalities are linear, \(P_{a,\ell}\) is a polytope.

Introduce a nonnegative joint-mass variable
\[
    q_{a,\ell,\omega}\ge 0,
    \qquad
    (a,\ell)\in\mathcal T,
    \ \omega\in\Omega .
\]
The intended meaning is the joint probability that the state is \(\omega\) and
the sender sends a refined recommendation of type \((a,\ell)\).  Consider the
linear program
\[
\begin{aligned}
\max_{q\ge 0}\quad
&
\sum_{(a,\ell)\in\mathcal T}
\sum_{\omega\in\Omega}
q_{a,\ell,\omega}v(\omega,a)
\\
\text{s.t.}\quad
&
\sum_{(a,\ell)\in\mathcal T}q_{a,\ell,\omega}
=
\mu_0(\omega),
\qquad
\forall \omega\in\Omega,
\\
&
\sum_{\omega\in\Omega}
q_{a,\ell,\omega}
\bigl(c_{a,\ell,\omega}-c_{a',\ell',\omega}\bigr)
\ge 0,
\\
&\hspace{40mm}
\forall (a,\ell)\in\mathcal T,
\ \forall (a',\ell')\in\mathcal T .
\end{aligned}
\]
The first set of constraints is Bayes plausibility.  For a type \((a,\ell)\)
with positive mass
\[
    \lambda_{a,\ell}:=\sum_{\omega\in\Omega}q_{a,\ell,\omega}>0,
\]
the induced posterior is
\[
    \mu_{a,\ell}(\omega)
    :=
    \frac{q_{a,\ell,\omega}}{\lambda_{a,\ell}}.
\]
The second set of constraints is exactly the homogeneous form of
\(\mu_{a,\ell}\in P_{a,\ell}\).  Therefore every feasible solution induces an
incentive-compatible signaling scheme: for every positive-mass type
\((a,\ell)\), send one signal with probability \(\lambda_{a,\ell}\), posterior
\(\mu_{a,\ell}\), and recommendation \(a\).  The sender value of this scheme is
precisely the LP objective.

Conversely, take any incentive-compatible signaling scheme.  For each signal
\(s\), let \(\lambda_s\) be its probability, \(\mu_s\) its posterior, and
\(a_s\) its recommended action.  Since the scheme is incentive compatible,
\[
    \rho(\mu_s,a_s)\ge \rho(\mu_s,a')
    \qquad
    \forall a'\in\mathcal A .
\]
Choose an active facet \(\ell_s\in\mathcal L_{a_s}\) such that
\[
    \rho(\mu_s,a_s)=\langle c_{a_s,\ell_s},\mu_s\rangle .
\]
Then for every \((a',\ell')\in\mathcal T\),
\[
    \langle c_{a_s,\ell_s},\mu_s\rangle
    =
    \rho(\mu_s,a_s)
    \ge
    \rho(\mu_s,a')
    \ge
    \langle c_{a',\ell'},\mu_s\rangle .
\]
Thus \(\mu_s\in P_{a_s,\ell_s}\).  Adding the joint mass
\[
    \lambda_s\mu_s(\omega)
\]
to the variable \(q_{a_s,\ell_s,\omega}\), and aggregating over signals with
the same action--facet type, gives a feasible LP solution with the same sender
value.  Hence the LP optimum equals the optimal persuasion value.

The number of variables is \(mL\).  The Bayes-plausibility constraints
contribute \(m\) equalities, and the incentive comparisons contribute at most
\(L^2\) inequalities.  Since all coefficient vectors are explicitly listed,
the coefficient bit length is polynomial in the input size.  Standard
polynomial-time linear programming algorithms therefore solve the problem in
time polynomial in \(m\), \(L\), and the input bit length.
\end{proof}

\section{Detailed Proof of Theorem \ref{thm:poly_time}}
\label{app:poly_proof}

In this appendix, we provide the rigorous constructive proof for the polynomial-time solvability of the persuasion problem when the state space size $|\Omega|=m$ is fixed. Write \(n=|\mathcal A|\).
\subsection{Geometric Decomposition of the Belief Simplex}
Recall that for each action $a \in \mathcal{A}$, the CVaR utility is piecewise linear:\begin{equation}\rho(\mu, a) = \max_{l \in \{1, \dots, m\}} \langle \mathbf{c}_{a,l}, \mu \rangle\end{equation}where $\mathbf{c}_{a,l} \in \mathbb{R}^m$ 
are constant vectors derived from utility parameters.The Receiver prefers action $a$ over $a'$ at belief $\mu$ if:\begin{equation}\max_{l} \langle \mathbf{c}_{a,l}, \mu \rangle \ge \max_{k} \langle \mathbf{c}_{a',k}, \mu \rangle\end{equation}
This condition holds if there exists an index $l$ such that for all $k$, $\langle \mathbf{c}_{a,l}, \mu \rangle \ge \langle \mathbf{c}_{a',k}, \mu \rangle$.
We define the set of \textbf{critical hyperplanes} $\mathcal{H}$ as the set of all loci where any two linear segments (either from the same action or different actions) intersect:\begin{equation}\mathcal{H} = \left\{ H_{a,l, a',k} \mid a, a' \in \mathcal{A}; l, k \in \{1, \dots, m\} \right\}\end{equation}
where $H_{a,l, a',k} = \{ \mu \in \mathbb{R}^m \mid \langle \mathbf{c}_{a,l} - \mathbf{c}_{a',k}, \mu \rangle = 0 \}$.
The total number of linear segments across all actions is $N_{seg} \le mn$. 
The number of pairwise comparisons (hyperplanes) is bounded by the square of this number:
\begin{equation}|\mathcal{H}| \le (mn)^2 = n^2m^2\end{equation}
Consider the arrangement $\mathcal{A}(\mathcal{H})$ formed by these hyperplanes restricted to the $(m-1)$-dimensional simplex $\Delta(\Omega)$. 
A classic result in combinatorial geometry states that an arrangement of $N$ hyperplanes in $\mathbb{R}^d$ partitions the space into at most $O(N^d)$ connected regions (cells).
Here $d = m-1$ and $N = n^2m^2$. Thus, the number of cells $K$ satisfies:\begin{equation}K = O\left( (n^2m^2)^{m-1} \right) = O\left( n^{2m-2} \right)\end{equation}
Since $m$ is a constant, $K$ is polynomial in \(n\).

\subsection{Algorithm and Complexity Analysis}
Within each open cell $C$ of the arrangement, the ordering of all linear functions $\{ \langle \mathbf{c}_{a,l}, \mu \rangle \}$ is fixed. This implies: For every action $a$, the active linear segment $l^*(a)$ (where $\rho(\mu, a)$ attains its maximum) is constant.
The Receiver's optimal action set $A^*(C) = \operatorname{argmax}_{a} \rho(\mu, a)$ is constant for all $\mu \in C$.

\textbf{The Algorithm}

\begin{itemize}
\item Construct the Arrangement: Enumerate all $O(n^{2m})$ cells. This can be done computationally by traversing the graph of the arrangement.
\item Identify Receiver Behavior: For each cell $C_j$, determine the Receiver's best response action $a_j$.

\item Optimize Convex Decomposition:According to Carathéodory's theorem, the optimal signaling scheme requires at most $m$ posterior beliefs. We iterate over all combinations of $m$ cells $\{C_{j_1}, \dots, C_{j_m}\}$ from the arrangement.
For each combination, we solve the following Linear Program (LP):
\begin{align*}\max_{\lambda, \mu_1, \dots, \mu_m} \quad & \sum_{k=1}^m \lambda_k v(\mu_k, a_{j_k}) \text{s.t.} \quad & \sum_{k=1}^m \lambda_k \mu_k = \mu_0 \& \mu_k \in \text{closure}(C_{j_k}), \quad \forall k \& \sum \lambda_k = 1, \quad \lambda_k \ge 0\end{align*}
(Note: The constraint $\mu_k \in C_{j_k}$ with variable weight $\lambda_k$ can be linearized by substituting $z_k = \lambda_k \mu_k$, where $z_k$ must lie in the conic hull of $C_{j_k}$).

\end{itemize}

The number of cell combinations is $\binom{K}{m} \approx K^m \approx (n^{2m})^m = n^{2m^2}$.For each combination, we solve a fixed-size LP.Thus, the total runtime is polynomial in \(n\), specifically $O(n^{2m^2})$. \qed

\section{Proof of Theorem~\ref{thm:succinct_polyhedral_np_hard}}
\label{app:succinct_polyhedral_np_hard}

\begin{proof}
We reduce from \textsc{$K$-Clique}.  Let \(G=(V,E)\) be an undirected graph
with \(|V|=n\), and let \(K\) be the target clique size.  The persuasion
instance has state space
\[
    \Omega=V,
\]
and the prior is uniform:
\[
    \mu_0(v)=\frac1n,
    \qquad v\in V.
\]
The receiver has two actions, denoted \(a_T\) and \(a_0\).  The sender obtains
payoff
\[
    v(v,a_T)=1,
    \qquad
    v(v,a_0)=0,
    \qquad \forall v\in V.
\]
Thus the sender wants to induce action \(a_T\).

The receiver's risk value for the safe action \(a_0\) is constant:
\[
    \rho(\mu,a_0)=1
    =\langle \mathbf 1,\mu\rangle .
\]
The receiver's risk value for the target action \(a_T\) is defined by a
succinct max-affine representation.  Let \(\mathcal C_K(G)\) denote the set of
all \(K\)-cliques of \(G\).  For every \(C\in\mathcal C_K(G)\), define the
affine coefficient vector \(c_C\in\{0,1\}^{V}\) by
\[
    c_C(v)=\mathbf 1\{v\in C\}.
\]
Set
\[
    \rho(\mu,a_T)
    =
    \max_{C\in\mathcal C_K(G)} \sum_{v\in C}\mu(v),
\]
with the convention that the maximum over an empty family is \(0\).  This is a
polyhedral risk functional: it is the maximum of affine functions of the
posterior.  The family is succinctly represented by the graph \(G\) and the
integer \(K\), rather than by explicitly listing all \(K\)-cliques.

We set the sender-value threshold to
\[
    \eta=\frac{K}{n}.
\]
We now prove that the graph contains a \(K\)-clique if and only if the sender
can achieve value at least \(\eta\).

First suppose that \(G\) contains a \(K\)-clique \(C\).  Consider the signaling
scheme that sends signal \(s_T\) exactly on states in \(C\), and sends signal
\(s_0\) otherwise.  Conditional on \(s_T\), the posterior is supported on
\(C\), hence
\[
    \rho(\mu_{s_T},a_T)
    \ge
    \sum_{v\in C}\mu_{s_T}(v)
    =1.
\]
Since \(\rho(\mu,a_0)=1\) for every posterior, action \(a_T\) is a weak best
response at \(\mu_{s_T}\).  Conditional on \(s_0\), recommend \(a_0\), which
is always a weak best response because \(\rho(\mu,a_T)\le 1=\rho(\mu,a_0)\).
The probability of signal \(s_T\) is \(K/n\), so the sender's expected value is
\[
    \frac{K}{n}\cdot 1 + \left(1-\frac{K}{n}\right)\cdot 0
    =
    \eta.
\]
Thus the persuasion instance has value at least \(\eta\).

Conversely, suppose that \(G\) has no \(K\)-clique.  Then
\(\mathcal C_K(G)=\emptyset\), and by construction
\[
    \rho(\mu,a_T)=0
    \qquad
    \text{for every posterior }\mu.
\]
Since \(\rho(\mu,a_0)=1\), action \(a_T\) is never incentive compatible at any
posterior.  Therefore every incentive-compatible signaling scheme recommends
\(a_0\) with probability one, and the sender's expected value is \(0<\eta\).
Hence the sender can achieve value at least \(\eta\) only if \(G\) contains a
\(K\)-clique.

The construction is polynomial in the size of \((G,K)\).  Therefore deciding
whether the sender can achieve value at least \(\eta\) is NP-hard.
\end{proof}

\section{Proofs for the Discretization Scheme}
\label{app:discretization_refinements}

Throughout this appendix, \(B_\rho=C_gL_\Psi\) denotes the finite-statistic
scale from Section~\ref{sec:qptas}.  In the CVaR specialization,
\(D\le mn\) and \(B_\rho=C_R=M_R(1+2/r)\), where \(n=|\mathcal A|\).
The proofs use the finite-statistic cell-access model of
Assumption~\ref{assump:finite_statistic_access}; when the statistics are
implicit, the displayed cell constraints are understood through the
corresponding separation oracle.

This appendix gives the proofs for Section~\ref{sec:qptas}. Appendices~\ref{app:qptas_soundness_proof} and \ref{app:qptas_completeness_proof} prove the LP soundness and completeness lemmas; 
Appendices~\ref{app:boundary_instability_proof} and \ref{app:strict_ic_margin_proof} prove the strict-IC margin claims. Throughout, $\bar\mu$ denotes the empirical $k$-uniform posterior sampled from $\mu$.

\subsection{Proof of Lemma~\ref{lem:qptas_soundness}}
\label{app:qptas_soundness_proof}
\begin{proof}
Let \(\varphi\) be feasible for LP~\eqref{lp:qptas_state_contingent}.  The
constraints \eqref{lp:signal_rule} and \eqref{lp:nonnegative} define a valid
state-contingent signaling rule.  For a signal
\(\sigma=(\bar\mu_\sigma,a_\sigma)\), let
\[
    p_\sigma
    :=
    \sum_{\omega\in\Omega}\mu_0(\omega)\varphi(\omega,\sigma)
\]
be its unconditional probability.  If \(p_\sigma>0\), the induced posterior is
\[
    \mu_\sigma(\omega)
    =
    \frac{\mu_0(\omega)\varphi(\omega,\sigma)}{p_\sigma}.
\]
The lower and upper cell constraints imply, after dividing by \(p_\sigma\),
that for every statistic \(j\in[D]
\),
\[
    -\frac{\epsilon_R}{L_\Psi}
    \le
    \langle g_j,\mu_\sigma\rangle
    -
    \langle g_j,\bar\mu_\sigma\rangle
    \le
    \frac{\epsilon_R}{L_\Psi}.
\]
By the Lipschitz condition on \(\Psi_a\), this gives
\[
    |\rho(\mu_\sigma,a)-\rho(\bar\mu_\sigma,a)|
    \le \epsilon_R,
    \qquad
    \forall a\in\mathcal A.
\]
Since \(\sigma\in\widehat\Sigma\), the recommended action is
\(2\epsilon_R\)-optimal at the grid center:
\[
    \rho(\bar\mu_\sigma,a_\sigma)
    \ge
    \max_{a'\in\mathcal A}\rho(\bar\mu_\sigma,a')-2\epsilon_R.
\]
Therefore, for any competing action \(a'\),
\[
\begin{aligned}
    \rho(\mu_\sigma,a')
    &\le \rho(\bar\mu_\sigma,a')+\epsilon_R \\
    &\le \rho(\bar\mu_\sigma,a_\sigma)+3\epsilon_R \\
    &\le \rho(\mu_\sigma,a_\sigma)+4\epsilon_R.
\end{aligned}
\]
Thus the IC regret of the recommendation is at most \(4\epsilon_R\).  Signals
with zero probability do not affect the induced scheme.
\end{proof}

\subsection{Proof of Lemma~\ref{lem:qptas_completeness}}
\label{app:qptas_completeness_proof}

\begin{proof}
Let \(\pi^*\) be a finite-support exact-IC signaling scheme attaining value
\(OPT\).  Index its signals by \(j\in J\).  Let \(p_j\) be the probability of
signal \(j\), \(\mu_j^*\) its posterior, and \(a_j\) its recommended action.
By exact incentive compatibility,
\[
    \rho(\mu_j^*,a_j)
    \ge
    \rho(\mu_j^*,a),
    \qquad
    \forall a\in\mathcal A.
\]
For each posterior \(\mu_j^*\), apply Lemma~\ref{lemma:approx} to choose a
\(k\)-uniform posterior \(\bar\mu_j\in\mathcal D_k\) such that
\[
    \max_a |\rho(\mu_j^*,a)-\rho(\bar\mu_j,a)|
    \le \epsilon_R.
\]
Then \((\bar\mu_j,a_j)\in\widehat\Sigma\), since for every action \(a\),
\[
\begin{aligned}
    \rho(\bar\mu_j,a_j)
    &\ge \rho(\mu_j^*,a_j)-\epsilon_R \\
    &\ge \rho(\mu_j^*,a)-\epsilon_R \\
    &\ge \rho(\bar\mu_j,a)-2\epsilon_R.
\end{aligned}
\]
Define a candidate LP solution by aggregating all original signals mapped to
the same discretized label:
\[
    \varphi^*(\omega,\sigma)
    :=
    \sum_{j:(\bar\mu_j,a_j)=\sigma}
    \pi^*(j\mid \omega).
\]
The signaling-rule and nonnegativity constraints are immediate from the fact
that \(\pi^*\) is a signaling scheme.

It remains to verify the cell constraints.  Fix
\(\sigma=(\bar\mu,a_\sigma)\in\widehat\Sigma\), and let
\[
    J_\sigma=
    \{j\in J:(\bar\mu_j,a_j)=\sigma\}.
\]
For every \(j\in J_\sigma\), Lemma~\ref{lemma:approx} gives
\[
    \left|
    \langle g_l,\mu_j^*-\bar\mu\rangle
    \right|
    \le
    \frac{\epsilon_R}{L_\Psi},
    \qquad
    \forall l\in[D].
\]
Multiplying the lower inequality by \(p_j\), using
\[
    p_j\mu_j^*(\omega)=\mu_0(\omega)\pi^*(j\mid\omega),
\]
and summing over \(j\in J_\sigma\) yields
\[
    \sum_{\omega\in\Omega}
    \mu_0(\omega)\varphi^*(\omega,\sigma)
    \left(
        g_l(\omega)-\langle g_l,\bar\mu\rangle
        +\frac{\epsilon_R}{L_\Psi}
    \right)
    \ge 0.
\]
This is exactly the lower cell constraint.  The upper cell constraint follows
from the upper inequality in the same way.  Hence \(\varphi^*\) is feasible.

Finally, the LP objective at \(\varphi^*\) equals the sender value of the
original scheme, because the recommended action attached to each discretized
label is the same action used by the original signal mapped to that label:
\[
\begin{aligned}
    &\sum_{\omega\in\Omega}\mu_0(\omega)
      \sum_{\sigma\in\widehat\Sigma}
      \varphi^*(\omega,\sigma)v(\omega,a_\sigma) \\
    &=
      \sum_{\omega\in\Omega}\mu_0(\omega)
      \sum_{j\in J}
      \pi^*(j\mid\omega)v(\omega,a_j)
      =OPT.
\end{aligned}
\]
Thus the LP optimum is at least \(OPT\).
\end{proof}

\subsection{Proof of Proposition~\ref{prop:boundary_instability}}
\label{app:boundary_instability_proof}
\begin{proof}
Fix any competing action $a'\neq a$. By the proxy condition,
\begin{equation}
\rho(\bar\mu,a)\ge \rho(\mu,a)-\epsilon,
\qquad
\rho(\bar\mu,a')\le \rho(\mu,a')+\epsilon.
\end{equation}
Therefore
\begin{equation}
\rho(\bar\mu,a)-\rho(\bar\mu,a')
\ge \rho(\mu,a)-\rho(\mu,a')-2\epsilon.
\end{equation}
Taking the minimum over $a'\neq a$ on the left-hand side is equivalent to subtracting the maximum competing value, so
\begin{equation}
\Gamma(\bar\mu,a)
=\min_{a'\neq a}\{\rho(\bar\mu,a)-\rho(\bar\mu,a')\}
\ge
\min_{a'\neq a}\{\rho(\mu,a)-\rho(\mu,a')\}-2\epsilon
=\Gamma(\mu,a)-2\epsilon.
\end{equation}
\end{proof}

\subsection{Proof of Theorem~\ref{thm:strict_ic_margin}}
\label{app:strict_ic_margin_proof}

\begin{proof}
Let
$
    \pi^\star_{\mathrm{strict}}
    =
    \{(p_j^\star,\mu_j^\star,a_j^\star)\}_{j\in J}
$
be the strict-margin benchmark scheme. By assumption,
$
    \Gamma(\mu_j^\star,a_j^\star)\ge\gamma
$
for every supported signal $j$.

For each $j$, apply Lemma~\ref{lemma:approx} to choose
$\bar\mu_j\in\mathcal D_k$ such that
$
    \max_{a\in\mathcal A}
    |\rho(\mu_j^\star,a)-\rho(\bar\mu_j,a)|
    \le \epsilon_R.
$
By Proposition~\ref{prop:boundary_instability},
\begin{equation}
    \Gamma(\bar\mu_j,a_j^\star)
    \ge
    \Gamma(\mu_j^\star,a_j^\star)-2\epsilon_R
    \ge
    \gamma-2\epsilon_R.
\end{equation}
Since $\epsilon_R\le\gamma/4$, we have
$
    \Gamma(\bar\mu_j,a_j^\star)\ge \gamma/2.
$
Thus
$
    (\bar\mu_j,a_j^\star)\in\widehat\Sigma_\gamma.
$

Now construct a feasible solution to the margin-filtered LP by the same
state-contingent merging argument used in Lemma~\ref{lem:qptas_completeness}.
Because all mapped pairs belong to $\widehat\Sigma_\gamma$, the construction is
feasible for the filtered LP. Since the recommended action $a_j^\star$ is
preserved for every original signal, the LP objective value of the constructed
solution is exactly $OPT_{\mathrm{strict}}$. Therefore the filtered LP optimum
is at least $OPT_{\mathrm{strict}}$.

It remains to prove strict incentive compatibility of the LP-induced scheme.
Let $\varphi$ be any feasible solution of the filtered LP and consider a
supported signal $\sigma=(\bar\mu_\sigma,a_\sigma)$. Its induced posterior is
$\mu_\sigma$. The LP cell constraints imply
$
    \max_{a\in\mathcal A}
    |\rho(\mu_\sigma,a)-\rho(\bar\mu_\sigma,a)|
    \le \epsilon_R.
$
Since $\sigma\in\widehat\Sigma_\gamma$,
$
    \Gamma(\bar\mu_\sigma,a_\sigma)\ge \gamma/2.
$
Applying Proposition~\ref{prop:boundary_instability} again gives
\begin{equation}
    \Gamma(\mu_\sigma,a_\sigma)
    \ge
    \Gamma(\bar\mu_\sigma,a_\sigma)-2\epsilon_R
    \ge
    \frac{\gamma}{2}-2\epsilon_R.
\end{equation}
Because $\epsilon_R<\gamma/4$, the right-hand side is strictly positive. Hence
$a_\sigma$ is a strict best response at $\mu_\sigma$ for every supported
signal.

Finally, the running time is the same as in Theorem~\ref{thm:qptas}, with
\(\epsilon_R=\Theta(\gamma)\). Hence
\[
    k=O\left(\frac{B_\rho^2\log D}{\gamma^2}\right),
\]
and the running time is
\[
    m^{O(B_\rho^2\log D/\gamma^2)}
    \operatorname{poly}(m,n,D,T_{\mathrm{eval}},T_{\mathrm{cell}})
\]
in the explicitly listed-statistic case, with the analogous separated-oracle
bound under statistic-cell separation.
\end{proof}

\section{Proofs for the Local Active-Facet Refinement}
\label{app:local_refinement_rewrite}

This appendix gives the details for the local refinement in
Section~\ref{subsec:local_refinements}.  Throughout this appendix the receiver's
risk value is max-affine:
\[
    \rho(\mu,a)=\max_{\ell\in\mathcal L_a}\langle c_{a,\ell},\mu\rangle .
\]
This includes CVaR and, more generally, explicitly listed polyhedral risk
preferences.  Fix a posterior region \(\mathcal U\subseteq\Delta(\Omega)\) and
radius \(\eta>0\).  Let
\[
    \mathcal U^\eta
    :=
    \{\nu\in\Delta(\Omega):\operatorname{dist}_1(\nu,\mathcal U)\le \eta\}.
\]
The local active-facet family is
\[
    \mathcal F_{\rm loc}
    :=
    \left\{
    (a,\ell):
    \exists \nu\in\mathcal U^\eta
    \text{ such that }
    \rho(\nu,a)=\langle c_{a,\ell},\nu\rangle
    \right\},
    \qquad
    N_{\rm loc}:=|\mathcal F_{\rm loc}|.
\]
We also write
\[
    C_{\rm loc}:=\\max_{(a,\ell)\in\mathcal F_{\rm loc}}
    \|c_{a,\ell}\|_\infty .
\]
For CVaR, one may take
\(C_{\rm loc}\le C_R=M_R(1+2/r)\), where
\(M_R=\max_{\omega,a}|u(\omega,a)|\).

\subsection{A Local Proxy Lemma}
\label{app:local_proxy_proof}

The first lemma formalizes why only locally active facets matter.  If both
posteriors lie in the local neighborhood \(\mathcal U^\eta\), then the value of
an action is always realized by a facet in \(\mathcal F_{\rm loc}\).  Hence
controlling this smaller family suffices to control all receiver values locally.

\begin{lemma}[Local active-facet proxy]
\label{lem:local_proxy_rewrite}
Let \(\mu\in\mathcal U\) and \(\bar\mu\in\mathcal U^\eta\).  Suppose
\[
    \max_{(a,\ell)\in\mathcal F_{\rm loc}}
    |\langle c_{a,\ell},\bar\mu-\mu\rangle|
    \le \epsilon .
\]
Then
\[
    \max_{a\in\mathcal A}
    |\rho(\bar\mu,a)-\rho(\mu,a)|
    \le \epsilon .
\]
\end{lemma}

\begin{proof}
Fix an action \(a\).  Let \(\ell_{\bar\mu}\in\mathcal L_a\) be an active
facet for action \(a\) at \(\bar\mu\), so that
\[
    \rho(\bar\mu,a)=\langle c_{a,\ell_{\bar\mu}},\bar\mu\rangle .
\]
Since \(\bar\mu\in\mathcal U^\eta\), the pair \((a,\ell_{\bar\mu})\) belongs
to \(\mathcal F_{\rm loc}\).  Therefore
\[
\begin{aligned}
    \rho(\bar\mu,a)-\rho(\mu,a)
    &=
    \langle c_{a,\ell_{\bar\mu}},\bar\mu\rangle
    -
    \max_{\ell\in\mathcal L_a}\langle c_{a,\ell},\mu\rangle \\
    &\le
    \langle c_{a,\ell_{\bar\mu}},\bar\mu\rangle
    -
    \langle c_{a,\ell_{\bar\mu}},\mu\rangle \\
    &=
    \langle c_{a,\ell_{\bar\mu}},\bar\mu-\mu\rangle
    \le \epsilon .
\end{aligned}
\]
Conversely, let \(\ell_{\mu}\) be active for action \(a\) at \(\mu\).  Since
\(\mu\in\mathcal U\subseteq\mathcal U^\eta\), we have
\((a,\ell_\mu)\in\mathcal F_{\rm loc}\).  Then
\[
\begin{aligned}
    \rho(\mu,a)-\rho(\bar\mu,a)
    &\le
    \langle c_{a,\ell_\mu},\mu\rangle
    -
    \langle c_{a,\ell_\mu},\bar\mu\rangle \\
    &=
    \langle c_{a,\ell_\mu},\mu-\bar\mu\rangle
    \le \epsilon .
\end{aligned}
\]
Combining the two inequalities gives
\(|\rho(\bar\mu,a)-\rho(\mu,a)|\le\epsilon\).  Taking the maximum over
\(a\in\mathcal A\) proves the claim.
\end{proof}

\subsection{Local Uniform Approximation}
\label{app:localized_union_proof}

The next lemma is the local analogue of the global uniform approximation
lemma.  It uses a union bound only over \(\mathcal F_{\rm loc}\).

\begin{lemma}[Local concentration over active facets]
\label{lem:localized_union_rewrite}
Fix \(\mu\in\mathcal U\), and let \(\bar\mu\) be the empirical distribution of
\(k\) i.i.d. samples from \(\mu\).  With probability at least \(1-\delta\),
\[
    \max_{(a,\ell)\in\mathcal F_{\rm loc}}
    |\langle c_{a,\ell},\bar\mu-\mu\rangle|
    \le \epsilon
\]
provided
\[
    k
    \ge
    \frac{2C_{\rm loc}^2}{\epsilon^2}
    \log\left(\frac{2N_{\rm loc}}{\delta}\right).
\]
Consequently, on the same event, if \(\bar\mu\in\mathcal U^\eta\), then
\[
    \max_{a\in\mathcal A}|\rho(\bar\mu,a)-\rho(\mu,a)|\le \epsilon .
\]
\end{lemma}

\begin{proof}
Fix \((a,\ell)\in\mathcal F_{\rm loc}\).  For samples
\(X_1,\ldots,X_k\sim\mu\), set
\(Y_i=c_{a,\ell}(X_i)\).  Then
\[
    \frac1k\sum_{i=1}^kY_i=\langle c_{a,\ell},\bar\mu\rangle,
    \qquad
    \mathbb E[Y_i]=\langle c_{a,\ell},\mu\rangle,
\]
and \(|Y_i|\le C_{\rm loc}\).  Hoeffding's inequality gives
\[
    \Pr\left(
    |\langle c_{a,\ell},\bar\mu-\mu\rangle|>\epsilon
    \right)
    \le
    2\exp\left(-\frac{k\epsilon^2}{2C_{\rm loc}^2}\right).
\]
Taking a union bound over \(N_{\rm loc}\) locally active facets yields
\[
    \Pr\left(
    \max_{(a,\ell)\in\mathcal F_{\rm loc}}
    |\langle c_{a,\ell},\bar\mu-\mu\rangle|>\epsilon
    \right)
    \le
    2N_{\rm loc}\exp\left(-\frac{k\epsilon^2}{2C_{\rm loc}^2}\right).
\]
The stated lower bound on \(k\) makes this probability at most \(\delta\).
The final claim follows from Lemma~\ref{lem:local_proxy_rewrite}.
\end{proof}

\subsection{A Variance-Sensitive Local Bound}
\label{app:bernstein_active_proof}

For some instances, the locally active affine pieces have much smaller variance
than their worst-case range.  The following Bernstein version replaces the
range-only sample size by a local variance term.

Define
\[
    V_{\rm loc}
    :=
    \sup_{\mu\in\mathcal U}
    \max_{(a,\ell)\in\mathcal F_{\rm loc}}
    \operatorname{Var}_{\omega\sim\mu}[c_{a,\ell}(\omega)] .
\]

\begin{lemma}[Variance-sensitive local active-facet bound]
\label{lem:bernstein_active_rewrite}
Fix \(\mu\in\mathcal U\), and let \(\bar\mu\) be the empirical distribution of
\(k\) i.i.d. samples from \(\mu\).  There is a universal constant \(C>0\) such
that, with probability at least \(1-\delta\),
\[
    \max_{(a,\ell)\in\mathcal F_{\rm loc}}
    |\langle c_{a,\ell},\bar\mu-\mu\rangle|
    \le \epsilon
\]
provided
\[
    k
    \ge
    C\max\left\{
    \frac{V_{\rm loc}}{\epsilon^2},
    \frac{C_{\rm loc}}{\epsilon}
    \right\}
    \log\left(\frac{2N_{\rm loc}}{\delta}\right).
\]
Consequently, if \(\bar\mu\in\mathcal U^\eta\), then
\(\max_a|\rho(\bar\mu,a)-\rho(\mu,a)|\le\epsilon\) on the same event.
\end{lemma}

\begin{proof}
Fix \((a,\ell)\in\mathcal F_{\rm loc}\) and again set
\(Y_i=c_{a,\ell}(X_i)\).  The random variables are bounded by
\(C_{\rm loc}\), and their variance is at most \(V_{\rm loc}\) by definition.
Bernstein's inequality gives, for a universal constant \(C_0>0\),
\[
    \Pr\left(
    |\langle c_{a,\ell},\bar\mu-\mu\rangle|>\epsilon
    \right)
    \le
    2\exp\left(
    -C_0 k\min\left\{
    \frac{\epsilon^2}{V_{\rm loc}},
    \frac{\epsilon}{C_{\rm loc}}
    \right\}
    \right).
\]
A union bound over \(\mathcal F_{\rm loc}\) shows that the desired event holds
with probability at least \(1-\delta\) whenever
\[
    k
    \ge
    C\max\left\{
    \frac{V_{\rm loc}}{\epsilon^2},
    \frac{C_{\rm loc}}{\epsilon}
    \right\}
    \log\left(\frac{2N_{\rm loc}}{\delta}\right)
\]
for a sufficiently large universal constant \(C\).  The risk-value conclusion
again follows from Lemma~\ref{lem:local_proxy_rewrite}.
\end{proof}

\subsection{Proof of the Certified Local Refinement}
\label{app:local_complexity_proof}

We now prove the certified local theorem stated in
Section~\ref{subsec:local_refinements}.  The statement assumes that the
benchmark posteriors lie in \(\mathcal U\), that the reduced LP uses grid
centers in \(\mathcal U^\eta\), and that the induced posteriors are certified
to remain in \(\mathcal U^\eta\).  These conditions ensure that all active
facets relevant to the benchmark and to the LP solution belong to
\(\mathcal F_{\rm loc}\).

\begin{proof}[Proof of Theorem~\ref{thm:local_complexity}]
The proof is the same soundness-completeness argument as in the global
discretization theorem, with the global statistic family replaced by
\(\mathcal F_{\rm loc}\).

\emph{Soundness.}  Consider any feasible solution of the reduced
state-contingent LP.  By assumption, every supported induced posterior
\(\mu_\sigma\) lies in \(\mathcal U^\eta\), and its associated grid center
\(\bar\mu_\sigma\) also lies in \(\mathcal U^\eta\).  The local cell
constraints guarantee
\[
    \max_{(a,\ell)\in\mathcal F_{\rm loc}}
    |\langle c_{a,\ell},\mu_\sigma-\bar\mu_\sigma\rangle|
    \le \epsilon_R .
\]
By Lemma~\ref{lem:local_proxy_rewrite},
\[
    \max_{a\in\mathcal A}
    |\rho(\mu_\sigma,a)-\rho(\bar\mu_\sigma,a)|
    \le \epsilon_R .
\]
If the signal label \((\bar\mu_\sigma,a_\sigma)\) is chosen from the same
approximate-center alphabet as in Theorem~\ref{thm:qptas}, then
\[
    \rho(\bar\mu_\sigma,a_\sigma)
    \ge
    \max_{a'}\rho(\bar\mu_\sigma,a')-2\epsilon_R .
\]
Combining these inequalities gives
\[
    \max_{a'}\rho(\mu_\sigma,a')-
ho(\mu_\sigma,a_\sigma)
    \le 4\epsilon_R .
\]
Thus the reduced LP has the same IC-regret soundness guarantee as the global
LP.

\emph{Completeness.}  Let a benchmark scheme be supported on
\(\mathcal U\).  For every benchmark posterior \(\mu\), Lemma
\ref{lem:localized_union_rewrite}, or the Bernstein version
Lemma~\ref{lem:bernstein_active_rewrite}, gives a grid posterior
\(\bar\mu\in\mathcal D_k\) such that
\[
    \max_{(a,\ell)\in\mathcal F_{\rm loc}}
    |\langle c_{a,\ell},\mu-\bar\mu\rangle|
    \le \epsilon_R .
\]
When the local grid is restricted to \(\mathcal D_k\cap\mathcal U^\eta\), we
choose such a proxy within \(\mathcal U^\eta\), as required by the theorem's
certification condition.  Lemma~\ref{lem:local_proxy_rewrite} then implies
uniform approximation of all receiver values at the benchmark posteriors.
Therefore any exact-IC benchmark action becomes an approximate-center label at
its local proxy.  Assigning the benchmark signal masses to these proxy labels
constructs a feasible solution of the reduced LP with the same sender value as
the benchmark, exactly as in the completeness proof for the global LP.

The sample-size bound follows from Lemma~\ref{lem:bernstein_active_rewrite}.
Taking
\[
    k_{\rm loc}
    =
    O\left(
    \max\left\{
    \frac{V_{\rm loc}}{\epsilon^2},
    \frac{C_{\rm loc}}{\epsilon}
    \right\}
    \log N_{\rm loc}
    \right)
\]
up to logarithmic confidence factors yields the stated local grid size
\(m^{k_{\rm loc}}\).  The Hoeffding version follows by replacing
Lemma~\ref{lem:bernstein_active_rewrite} with
Lemma~\ref{lem:localized_union_rewrite}.  This proves the theorem.
\end{proof}

\section{Supplementary Numerical Experiments}
\label{app:experiments}

In this appendix, we provide detailed setups and analyses for the numerical experiments conducted to validate our theoretical findings. We explore three key dimensions: the performance gap between CVaR-aware and risk-neutral persuasion, the impact of risk aversion on information disclosure, and the computational efficiency of our approximation scheme.

\subsection{Implementation Details and Computational Environment}
\label{app:exp_environment}

All experiments were implemented in Python. The optimization-based experiments use Gurobi through the \texttt{gurobipy} interface. The CVaR-aware persuasion problems are encoded as mixed-integer linear programs: the continuous variables represent the joint probability of state and recommended action, and binary variables select the active linear piece of the CVaR representation. The risk-neutral benchmark in Figure~\ref{fig:exp_comparison} is solved as a linear program with standard expected-utility IC constraints. Gurobi output was disabled in all reported runs. For the large-instance approximation experiment in Figure~\ref{fig:exp_qptas}, the exact benchmark is computed with \texttt{MIPGap=0}; approximate runs use \texttt{MIPGap} equal to the plotted $\epsilon$ values and a per-solve time limit of 120 seconds. The big-$M$ constant used in the MILP implementation is $10.0$.

The deterministic experiments in Figures~\ref{fig:exp_comparison} and~\ref{fig:entropy} do not use random sampling. The discretization scheme experiment in Figure~\ref{fig:exp_qptas} generates synthetic instances with $|\Omega|=60$, $n=20$, $r=0.25$, and 20 trials for each $\epsilon$. The instance generator supports a NumPy seed argument; in the script used for the reported figure, the trial loop records the intended seed schedule $42+t$ for trial $t$, but the generator call leaves the seed unset. As a result, rerunning the script reproduces the experimental protocol but may not reproduce the plotted points exactly. A fully deterministic rerun can be obtained by passing the recorded seed $42+t$ to the generator in trial $t$.

The experiments were run on a MacBook Pro with an Apple M1 Pro processor, 10 CPU cores, and 32GB memory, using macOS 26.4.1. The software environment used Python 3.9.13, NumPy 1.26.4, SciPy 1.13.1, pandas 1.4.4, Matplotlib 3.5.2, and Gurobi/\texttt{gurobipy} 12.0.3. The reported discretization scheme runtime curve records wall-clock time measured around each Gurobi solve; the plotted values aggregate the mean runtime and mean relative error across the 20 trials, with standard-deviation error bars.

\subsection{Performance Comparison: CVaR vs. Risk-Neutral Senders}
\label{app:exp_comparison}

\textbf{Experimental Setup.} We evaluate the Sender's expected utility under two distinct signaling strategies: our proposed CVaR-aware mechanism versus the standard risk-neutral (expected utility) persuasion model. We vary the Receiver's risk tolerance level, parameterized by the quantile $r \in (0, 1)$, to observe performance robustness.
\begin{itemize}
    \item \textbf{Scenario 1 ($2 \times 2$ Case):} A minimal setting with two states (Bad/Good) and two actions. The "Safe" action yields a constant utility of 0.4. The "Risky" action yields 0 in the Bad state and 1 in the Good state (Expected Value = 0.5).
    \item \textbf{Scenario 2 ($5 \times 3$ Case):} A complex setting with 5 states and 3 actions. Action $a_0$ is safe. Action $a_1$ is highly risky (20\% probability of catastrophic zero return, 80\% probability of high return). Action $a_2$ serves as an intermediate noise option.
\end{itemize}

\textbf{Results and Analysis.} 

Figure \ref{fig:exp_comparison} (see main text or placeholder below) illustrates the Sender's utility as a function of $r$.
\begin{itemize}
    \item \textbf{Utility Collapse in Standard Models:} In the $2 \times 2$ case, the standard risk-neutral model (red dashed line) blindly mixes Good and Bad states to maximize expected value. When $r$ is low (high risk aversion), this mixed signal fails to satisfy the Receiver's CVaR constraint, leading the Receiver to reject the risky action entirely. Consequently, the Sender's utility collapses to 0.
    \item \textbf{Robustness of CVaR Model:} The CVaR-aware model (blue solid line) successfully identifies the Receiver's "risk floor." At low $r$, it adopts a more conservative disclosure strategy to eliminate tail risks, thereby securing a non-zero utility.
    \item \textbf{Convergence:} As $r \to 1$, the Receiver approaches risk neutrality. The CVaR constraints relax, and our model's performance smoothly converges to the standard benchmark, demonstrating generalized applicability.
\end{itemize}

% Placeholder for Figure 1 (User will insert image)
\begin{figure}[htbp]
	\centering 
	\includegraphics[width=\textwidth]{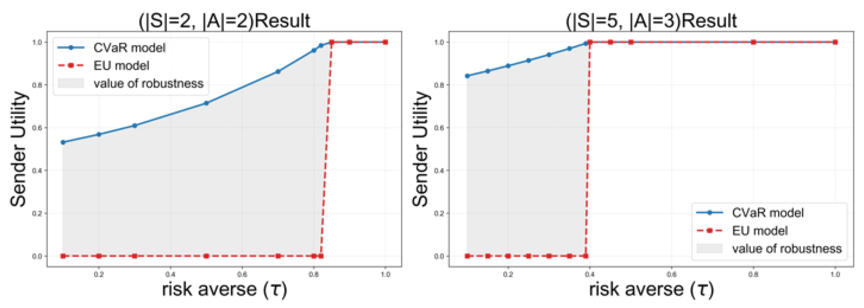}
    \caption{Comparison of Sender's utility under CVaR-aware vs. Standard Expected Utility models across varying risk tolerance levels $r$.}
    \label{fig:exp_comparison}
\end{figure}

\subsection{Impact of Risk Preference on Information Disclosure}
\label{app:exp_entropy}

\textbf{Experimental Setup.} To understand how risk aversion shapes the optimal information structure, we simulate a financial advisory scenario with $|\Omega|=5$ states and three actions:
\begin{enumerate}
    \item \textbf{Deposit:} Risk-free, constant utility 0.55.
    \item \textbf{Bond:} Low risk, linear return profile, expected utility 0.60.
    \item \textbf{Stock:} High risk, convex return profile (high upside, extreme tail risk of 0), expected utility 0.70.
\end{enumerate}
The Sender strictly prefers the Receiver to choose "Stock" ($v_{stock}=1, v_{others}<1$). We measure the degree of information disclosure using the \textbf{posterior entropy} $H(\pi) = \sum_{s} \mathbb{P}(s) \mathcal{H}(\mu_s)$. Lower entropy implies higher disclosure (precision), while higher entropy implies information pooling (ambiguity).

\textbf{Results and Analysis.} 
 Figure \ref{fig:entropy} (placeholder) reveals a monotonic negative correlation between risk tolerance and information precision:
\begin{itemize}
    \item \textbf{High Disclosure at Low $r$:} When the Receiver is extremely risk-averse, the Sender is forced to provide high-precision signals to "prove" that the tail risk has been eliminated. Consequently, posterior entropy is low.
    \item \textbf{Strategic Obfuscation at High $r$:} As risk tolerance increases, the Sender gains slack to mix "bad" states with "good" ones while still satisfying the CVaR constraint. This allows for partial pooling, leading to higher posterior entropy and higher Sender utility.
    \item \textbf{Step-wise Transitions:} The entropy curve exhibits discrete jumps, corresponding to the Receiver's optimal action switching from Deposit $\to$ Bond $\to$ Stock as persuasion becomes feasible.
\end{itemize}

% Placeholder for Figure 2
\begin{figure}[htbp]
    \centering 
    \includegraphics[width=0.6\textwidth]{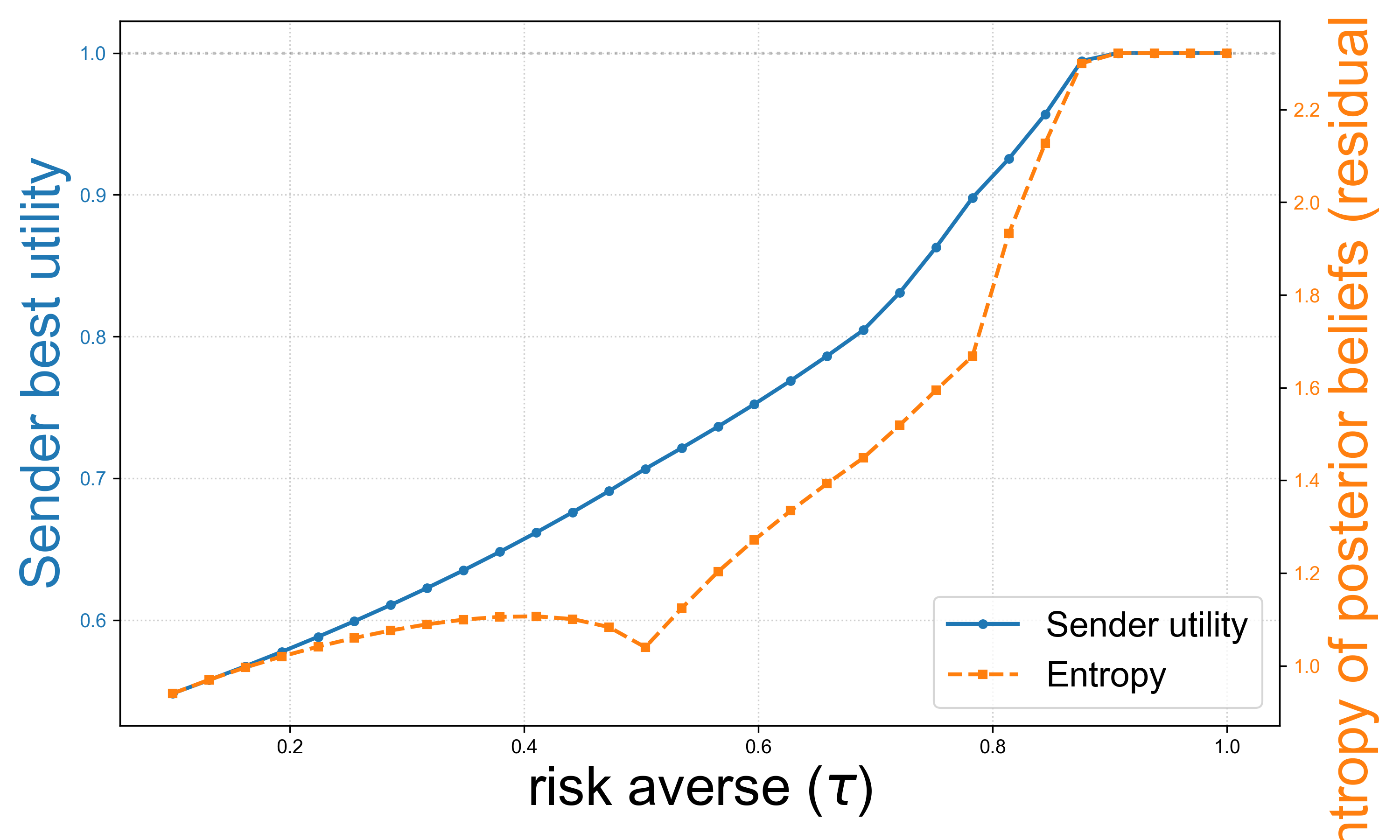}
    \caption{Evolution of posterior entropy (information ambiguity) with respect to risk tolerance $r$.}    \label{fig:entropy}
\end{figure}

\subsection{Finite-Precision Discretization: Accuracy-Time Trade-off}
\label{app:exp_qptas}

We evaluate the discretized implementation as a finite-precision approximation to the active-facet LP. The goal is not to replace the exact LP in the explicit finite-state model, but to illustrate how grid resolution affects sender value and IC violation when posterior beliefs are restricted to a finite representation.

\textbf{Experimental Setup.} 
\begin{itemize}
    \item \textbf{Scale:} $|\Omega|=60$ states, $n=20$ actions.
    \item \textbf{Adversarial Structure:} To challenge the algorithm, we construct "hard" instances where tail states have double probability weights (simulating frequent crises) and action variances grow exponentially ($ \sigma_a \propto e^{a/3}$), creating highly non-convex feasibility boundaries.
    \item \textbf{Parameters:} Risk level $r=0.25$. We vary the approximation parameter $\epsilon \in \{0.25, \dots, 0.0\}$.
\end{itemize}

\textbf{Results and Analysis.} 
 Figure \ref{fig:exp_qptas} (placeholder) plots the Relative Error (Orange) and Running Time (Blue) against $\epsilon$.
\begin{itemize}
    \item \textbf{High Precision:} The relative error decreases monotonically with $\epsilon$. Notably, at $\epsilon=0.05$, the empirical error is already below 2\%, suggesting the algorithm performs much better than the worst-case theoretical bound.
    \item \textbf{Computational Cost:} For coarse approximations ($\epsilon > 0.08$), the runtime is negligible (seconds). However, as $\epsilon \to 0$ (specifically $\epsilon < 0.03$), the runtime spikes exponentially. This confirms the theoretical complexity class of $n^{O(1/\epsilon^2)}$ and highlights the "expensive" nature of exact solutions in this non-convex landscape. The discretization scheme provides a tunable knob to navigate this trade-off effectively.
\end{itemize}

% Placeholder for Figure 3
\begin{figure}[h]
    \centering
    \includegraphics[width=0.6\textwidth]{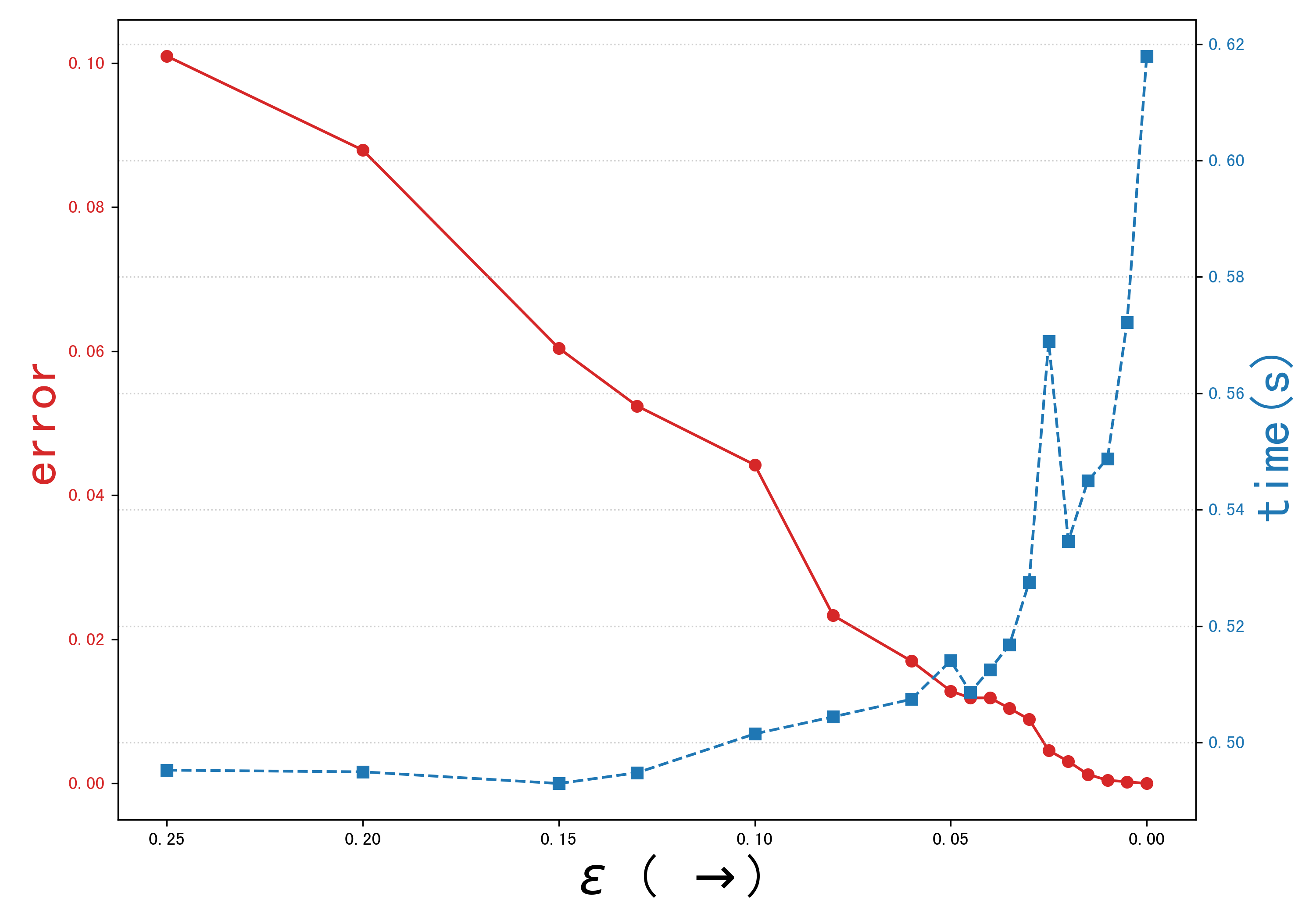}
    \caption{Trade-off between finite-precision accuracy and runtime for the posterior discretization scheme}
    \label{fig:exp_qptas}
\end{figure}

\end{document}